\documentclass[10pt,journal,cspaper,compsoc]{IEEEtran}
\ifCLASSOPTIONcompsoc
\else
\fi

\ifCLASSINFOpdf
\else
\fi

\usepackage{cite}      
\usepackage{graphicx}  
    
\usepackage{amssymb}
\usepackage{mathrsfs}
\usepackage{amsmath}   

\usepackage[sans]{dsfont}
\usepackage{stmaryrd}
\usepackage{pifont}
\usepackage{wasysym}
\usepackage{algorithm,algorithmicx,algpseudocode}
\usepackage{array}
\usepackage{url}
\usepackage[usenames]{color}
\usepackage[tight,footnotesize]{subfigure}
\usepackage{multirow}

\usepackage{flushend}

\newtheorem{example}{Example}
\newtheorem{definition}{Definition}
\newtheorem{lemma}{Lemma}
\newtheorem{theorem}{Theorem}

\def\endexam{\hspace*{\fill}~$\blacksquare$\par\endtrivlist\unskip}

\algnewcommand{\LeftComment}[1]{\Statex \(\triangleright\) #1}

\begin{document}
%
\title{Effective Static and Adaptive Carrier Sensing for Dense Wireless CSMA Networks}
%
%
%
%

\author{Chi-Kin Chau,~\IEEEmembership{Member,~IEEE,} Ivan W. H. Ho,~\IEEEmembership{Member,~IEEE,} Zhenhui Situ, \\ Soung Chang Liew,~\IEEEmembership{Fellow,~IEEE}, and Jialiang Zhang
\IEEEcompsocitemizethanks{\IEEEcompsocthanksitem C.-K. Chau is with the Masdar Institute of Science and Technology, UAE. E-mail: ckchau@masdar.ac.ae. I.W.-H. Ho and Z. Situ are with the Hong Kong Polytechnic University, Hong Kong. E-mail: ivanwh.ho@polyu.edu.hk, z.situ@connect.polyu.hk. S. C. Liew and J. Zhang are with the Chinese University of Hong Kong, Hong Kong. E-mail: {soung, jlzhang4}@ie.cuhk.edu.hk. Partial preliminary results have been presented in \cite{CCLZ13aggcs}.}
}%

\markboth{}%
{Chau \MakeLowercase{\textit{et al.}}: 
Effective Static and Adaptive Carrier Sensing Mechanisms for Dense CSMA Networks}
%


\IEEEcompsoctitleabstractindextext{%
\begin{abstract}
The increasingly dense deployments of wireless CSMA networks arising from applications of Internet-of-things call for an improvement to mitigate the interference among simultaneous transmitting wireless devices. For cost efficiency and backward compatibility with legacy transceiver hardware, a simple approach to address interference is by appropriately configuring the carrier sensing thresholds in wireless CSMA protocols, particularly in dense wireless networks. Most prior studies of the configuration of carrier sensing thresholds are based on a simplified conflict graph model, whereas this paper considers a realistic signal-to-interference-and-noise ratio model. We provide a comprehensive study for two effective wireless CSMA protocols: {\em Cumulative-interference-Power Carrier Sensing} and {\em Incremental-interference-Power Carrier Sensing}, in two aspects: (1) {\em static} approach that sets a universal carrier sensing threshold to ensure interference-safe transmissions regardless of network topology, and (2) {\em adaptive} approach that adjusts the carrier sensing thresholds dynamically based on the feedback of nearby transmissions. We also provide simulation studies to evaluate the starvation ratio, fairness, and goodput of our approaches. 
\end{abstract}

}

\maketitle

\IEEEdisplaynotcompsoctitleabstractindextext

%
\IEEEpeerreviewmaketitle

\section{Introduction}

The recent years have witnessed a dramatic rise in the deployments of increasingly dense wireless networks. In particular, the popular paradigm of Internet-of-things \cite{IOT1}\cite{IOT2} is accelerating the applications of pervasive portable wireless devices. These wireless devices are expected to be low-cost and ubiquitously present. Hence, a cost efficient and scalable communication mechanism is required to support these densely populated communicating devices.

Carrier-sensing multi-access (CSMA) protocols, a prominent class of distributed and randomized medium-access protocols, have been widely employed in various practical wireless networks (e.g., WiFi, Zigbee), which can empower Internet-of-things. The benefits CSMA include simple implementation and self-organizing adaptation to dynamic wireless network topology. 
Nonetheless, CSMA often suffers from considerably degraded performance in dense wireless networks due to the limitation of distributed decisions at transmitters, and the cumulative interference from surrounding simultaneous transmissions in the neighborhood. Improving the performance of CSMA, in particular for dense wireless networks, is a crucial research problem in wireless networking. 

For cost efficiency and backward compatibility with legacy transceiver hardware, one of the simplest approaches to address  interference in dense wireless networks is to appropriately configure the carrier sensing thresholds in wireless CSMA protocols. The {\em carrier sensing threshold} determines the ability of transmission permitted by the locally measured interference level. Since CSMA is a distributed protocol, transmitters determine their transmission operations based on locally perceived information. The locally measured interference level can assist the estimation of the likelihood of a successful reception at the receiver, and the interference that may cause to other surrounding simultaneous transmissions. The carrier sensing threshold is typically a pre-set parameter in the software driver of transceiver. Certain WiFi chipset vendors allow tunable carrier sensing threshold (as known as energy detection (ED) threshold) through manufacturer-provided driver interfaces \cite{ZGYCRH04cs}.

There has been prior work \cite{GCL06WirelessMultiHopThput, JL08hidden, DDT09csma} on interference mitigation among simultaneous transmissions in wireless CSMA networks, but most assumed an overly simplified interference model based on the notion of ``conflict graphs''. This model ignores the realistic physical layer characteristic of signal-to-interference-and-noise ratio (SINR) in wireless communications. Meanwhile, in the literature of physical-layer wireless communication, there has been extensive work \cite{PC1,PC2,VANET1,VANET2} utilizing power control or signal processing techniques to reduce the interference among simultaneous transmissions. These results usually require more sophisticated hardware designs or considerable revamps of physical layer protocols. In wireless networks populated with legacy hardware or low-end wireless transceivers, a simpler solution with minimum reconfigurations and redesigns of wireless systems is desirable.
In this work, we aim at providing practical solutions by appropriately configuring the carrier sensing thresholds in CSMA protocols to mitigate the interference among simultaneous transmissions in dense CSMA networks, which require minimal reconfigurations or low-cost implementations in the existing software stacks or drivers of transceivers.

We provide a comprehensive study for two effective CSMA mechanisms: Cumulative-interference-Power Carrier Sensing (CPCS) and Incremental-interference-Power Carrier Sensing (IPCS). CPCS is a conventional approach employed by CSMA, in which a transmitter proceeds its transmissions only if the locally measured cumulative interference power is below the carrier sensing threshold. In today's CSMA devices, the carrier sensing threshold is a pre-set fixed value, which is pre-configured to follow a certain rule of thumb, for example, 10-20 dB above the background noise level. However, there has been a lack of theoretical ground for configuring such a value. Meanwhile, IPCS is an alternative approach proposed by our previous work \cite{FLH10IPCS} to harness the detection of incremental power of interference as a basis for determining local transmissions. 

In this paper, we provide a theoretical basis for configuring a robust carrier sensing threshold for CPCS to ensure interference-safe transmissions (i.e., no transmission failure due to hidden nodes) no matter how dense the network topology becomes. In particular, we note that configuring a robust carrier sensing threshold is non-trivial, as it needs to consider the salient effects of arbitrary ordering of carrier sensing operations and the presence of ACK frames in the SINR model. On the other hand, we present adaptive approaches to both CPCS and IPCS that adjust the carrier sensing thresholds based on the dynamic feedback of nearby transmissions. Our adaptive approaches can significantly improve the goodput while balancing the fairness and starvation ratio in dense wireless networks. We provide extensive simulation studies to evaluate the performance of our approaches. Finally, we discuss the issues of practical implementation.

\subsection{Summary of Results and Contributions}

In CSMA networks, before a transmitter attempts its transmission, it needs to infer the channel condition. If the transmitter infers that its transmission is not interference-safe, namely possibly upsetting (or to be upset by) any on-going transmission, then it defers its transmission. 
There are two approaches of how the inference is determined at transmitters:
\begin{enumerate} 

\item ({\em CPCS}):
A conventional approach in the existing CSMA protocol (which we call {\em Cumulative-interference-Power Carrier-Sensing} (CPCS)) is to let the transmitter measure the cumulative interference power -- the total power of all concurrent transmissions and the background noise at the pending transmitter. A transmission will proceed only if the locally measured cumulative interference power is below a pre-specified {\em carrier sensing threshold}. 

\item ({\em IPCS}):
Another approach is recently proposed by \cite{FLH10IPCS}, called {\em incremental-Power Carrier-Sensing} (IPCS), which infers the distances between the concurrent transmitters based on the local measurement of incremental power changes at each node. Specifically, a transmitter will defer transmission if the incremental power change measured recently is above a power threshold. IPCS requires a modification to the software protocol stack, but not the hardware. 

\end{enumerate} 
CPCS and IPCS can be configured in two manners:
\begin{enumerate} 

\item ({\em Static}): The system parameters (e.g., carrier sensing thresholds) remain fixed throughout the operations of CSMA, regardless of changes in the number or locations of nodes. A static approach is suitable, if there is limited ability to dynamically adjust the carrier sensing threshold (e.g., in the existing CSMA systems). Although a static approach is not optimized to specific network topologies, it gives a sensible baseline setting, and can be applied to highly dense wireless networks.

\item ({\em Adaptive}): The adaptive mechanism allows the parameters to be adjustable according to dynamic changes of the environment. Instead of a uniform setting of static parameters at every node, the adaptive mechanism adapts heterogeneous parameters at different nodes. Unlike the static approach, the adaptive approach requires changes to the original protocols, and the availability of feedback. The initial values of parameters in adaptive mechanisms usually follow from the static setting. 

\end{enumerate} 

We note that an improperly high carrier sensing threshold fails to safe-guard interference, leading to the hidden node (HN) problem \cite{JL08hidden}, whereas an improperly low threshold causes overly conservative protection against interference and decrease in throughput, caused by the exposed node (EN) problem \cite{JL08hidden}. For configuring of the static carrier sensing thresholds, we focus on the robust value that can guarantee {\em interference-safe} transmissions among nodes (i.e., no transmission failure due to hidden nodes), and hence, is able to cope with highly dense wireless networks. We note that there is a trade-off between hidden nodes and exposed nodes. Our static mechanisms eliminate the hidden nodes, while the adaptive mechanisms improve the exposed nodes. 

Due to the following two effects, it is non-trivial to determine a robust threshold:
\begin{enumerate} 

\item ({\em Effect of Ordering}): CSMA is a distributed protocol, in which the transmitters decide their transmissions without global coordination. Without centralized coordination, the transmission order of the transmitters may be arbitrary. Specifically, an earlier transmitter that measured low interference power before the start of transmission may be disrupted by a later transmitter that causes unforeseen higher interference power to it. It is desirable to set a carrier sensing threshold to tolerate arbitrary ordering of local measurements of transmitters.

\item ({\em Effect of ACKs}): CSMA is an ACK-based protocol, in which the receivers are required to reply an ACK frame for each successful transmission. Hence, the carrier sensing threshold not only must ensure the successful receptions of DATA frames in one direction, but also the successful receptions of ACK frames in the opposite direction in the presence of other interfering transmitters. The consideration of bi-directional communications in terms of SINR complicates the configuration of carrier sensing threshold. 

\end{enumerate} 

In the following, we devise a static carrier sensing threshold for interference-safe CPCS to cope with both the effects of arbitrary ordering and ACKs under the SINR model. We also compare that with static setting of IPCS obtained in previous work \cite{FLH10IPCS}. 
Furthermore, we provide a performance evaluation of CSMA networks with CPCS, and observe that the carrier sensing threshold provided by our theoretical study is relatively robust in spite of uncertain parameters and the presence of fading. 

For adaptive mechanisms, we design adaptive CPCS and IPCS that adjust the carrier sensing thresholds starting at the values of static settings, with respect to cumulative power-interference power and pairwise incremental power, respectively. Note that adaptive mechanisms can strike a balance between hidden node problem and exposed node problem. Through simulation studies, we observe that adaptive mechanisms can achieve performance superior to static mechanisms, in terms of goodput, fairness, and failure rates (in the presence of fading).

\subsection{Related Work and Comparisons}

For static mechanisms of CSMA, ensuring interference-safe transmissions has been addressed in the extant literature (e.g., \cite{JL08hidden, DDT09csma, GBCG11, CGHT11}), in which a CSMA network was often modelled by a ``conflict graph'' that is induced by a geometric graph based on the transmitters and receivers. 
The conflict graph model relies on the binary constraint among pairs of transmitters, which does not consider the additive property of wireless signals. Hence, it calls for a more realistic model with signal-to-interference-and-noise ratio (SINR). The study of the hidden node problem in the SINR model has begun recently. In \cite{XGB02RTSCTS}, it reports that the common CTS/RTS mechanism, which relies on the assumption that the decodable range of a CTS/RTS message is comparable to the interfering range, is not sufficient to ensure interference-safe transmissions in the SINR model, because the sum of individually insignificant interference power can still be considerably large in the SINR model, and hence, a transmitter can affect very far-off nodes, other than those that can decode its packets. 
Two previous studies akin to this work are our prior work \cite{CCL09csma, FLH10IPCS}. In \cite{CCL09csma}, we proved the scaling law of capacity of CSMA in the SINR model. But the result relies on a simplified model of power-threshold based carrier sensing, which ignores the effect of ordering of local measurements. This work extends that study to consider the effect of ordering in CSMA. 

Previously, \cite{FLH10IPCS} proposed an alternate approach called Incremental-Power Carrier-Sensing (IPCS). In this paper, we compare the static IPCS and CPCS by evaluation and analysis. Note that for benchmarking IPCS with standard carrier sensing (i.e., CPCS), \cite{FLH10IPCS} also used the same carrier sensing threshold optimized for IPCS for CPCS. However, the carrier sensing threshold for IPCS may not be optimal for CPCS. In this paper, we provide a carrier sensing threshold appropriate for standard carrier sensing. Our solution of CPCS can be shown to have comparable performance to IPCS. 

There are other studies \cite{SCS03exposed, JL08hidden} to address a related problem of exposed nodes. We note that there is an inevitable trade-off between hidden node and exposed node problems \cite{JL08hidden}, as a solution addressing one of the problems often causes the other problem. Since there is a lack of thorough study for hidden node problem in a realistic setting of the SINR model, it is difficult to robustly address both problems simultaneously. Our results for static carrier sensing thresholds in this paper solve the former problem and provide a cornerstone for a complete solution for both problems in the SINR model. 

Previous investigations on adaptive mechanisms for wireless networks focused mostly on adjusting the transmission powers of nodes \cite{PC1,PC2}. Our paper here pursues the alternative of adjusting the carrier-sensing threshold of nodes. We note that our approach can maintain backward-compatibility with legacy hardware, when power control is limited (e.g., by the maximum allowable transmit power according to different regulations) or only allowing a small number of discrete levels of transmission power. Low-end devices that have limited form factor will benefit from adaptive carrier sensing mechanisms. In addition, one can achieve better network performance with the same power consumption through adaptively adjusting the carrier sensing threshold, which might not always be the case for power control.

There are prior results considered tuning the carrier sensing thresholds \cite{ZGYCRH04cs,FVV06cs,KLH06cs}, and \cite{SurveyACS} provided a comprehensive survey on adaptive carrier sensing. Some papers also evaluated the effectiveness of carrier-sensing control through experimental testbeds \cite{Testbed1, Testbed2}. \cite{ZGYCRH04cs} studied an upper bound for setting the carrier sensing threshold, which did not consider the effect of ordering and ACKs. \cite{ZGYCRH04cs} also proposed an adaptive carrier sensing threshold scheme that requires distribution of local measurements (SINR values, adaptation decision, etc.) through piggy backing in the ACK. In our adaptation algorithm, nodes do not need any estimation or measurement, they only need to check if they are deprived of opportunities to transmit or are not able to receive ACK for a sustained period. Thus, our algorithm is more direct and simple in terms of detection. \cite{MACOverhead} showed that the MAC overhead has important impact on the choice of carrier-sensing range, and \cite{MultiRate} considered other further factors on the choice of carrier-sensing range (e.g., channel rates, multi-hop forwarding, etc.). There are also studies that focus on topology-controlled networks \cite{TopologyCS} for adaptive physical carrier-sensing control. On the other hand,\cite{FVV06cs,KLH06cs} only considered maximizing the spatial reuse, rather than fairness and starvation ratio in the CSMA networks. In this paper, we jointly consider the goodput and fairness in the network and show that the proposed schemes can strive for a good balance between them for improving the overall network efficiency.

\section{Model and Notations} \label{sec:model}

To guarantee acceptable performance at each node, we provide carrier sensing threshold to guarantee interference-safe transmissions under the SINR model, by which each node can ensure that no interference is caused to other concurrent transmissions no matter how dense the network topology becomes. This provides a basis for further adaptive adjustment of the carrier sensing threshold.
This section presents the formal definition of carrier sensing threshold and hidden node problem under the SINR model.

Consider a set of links ${X}$. For each link $i \in X$, we let $t_i$ be the transmitter, and $r_i$ be the receiver. We also write $t_i$ and $r_i$ as the coordinates of transmitter and receiver, respectively. Let 
\begin{equation}
{\sf dist}(i,j) \triangleq \min(|t_j - r_i|, |r_j - t_i|, |r_j - r_i|, |t_j - t_i|) 
\end{equation}
which is the minimum distance among the transmitters and receivers between a pair of links $i, j$.

To capture the feasibility of interference-safe transmissions in the presence of possible ACKs under the SINR model, we define the following set of feasible states of concurrently transmitting links.

\medskip

\begin{definition}
({\em The set of bi-directional interference-safe feasible states ${\mathscr B}^{}_{\sf P, N}\big[{X},\beta\big] \subseteq 2^X$}): A subset of links ${S}$ are called ``interference-safe'' $(i.e., {S} \in {\mathscr B}^{}_{\sf P, N}\big[{X},\beta\big])$, if and only if for all $i \in {S}$, we have
\begin{equation} \label{eqn:B_ag}
 \frac{{\sf P} |t_i - r_i|^{-\alpha}}{{\sf N} + \underset{j \in {S} \backslash \{ i \}}{\sum} {\sf P} \cdot {\sf dist}(i,j)^{-\alpha}} \ge \beta 
 \end{equation}
 where ${\sf P}$ is the transmission power level, ${\sf N}$ is the background noise level, $\alpha$ is the path-loss coefficient and $\beta$ is the minimal SINR threshold to ensure successful decoding of a packet. 
\end{definition}

The above definition captures the setting of successful receptions of DATA frames not only in one direction, but also in the opposite direction for the successful receptions of ACK frames in the presence of other interfering transmitters.

The set of feasible states ${\mathscr B}^{}_{\sf P, N}\big[{X},\beta\big]$ generalizes the commonly used feasibility condition in the SINR model that only considers uni-directional communications. CSMA is supposed to enable the links in $X$ to operate within the constraint of ${\mathscr B}^{}_{\sf P, N}\big[{X},\beta\big]$, requiring no coordination among the links, by only local interference power measurement. That is, if the subsets of links allowed to transmit simultaneously by a carrier sensing mechanism are always within the set of feasible states ${\mathscr B}^{}_{\sf P, N}\big[{X},\beta\big]$, then the carrier sensing mechanism is called interference-safe (namely, with no hidden node problem).

\subsection{CPCS}

In this section, we first consider a naive definition of feasible states of carrier sensing that ignores the ordering effect by an illustrative example. Then, we provide a more robust definition that can accommodate the ordering effect.

\medskip

\begin{definition}
({\em The set of simple carrier sensing feasible states ${\mathscr C}^{\sf simp}_{\sf P, N}\big[{X},{\sf t}_{\rm cs}\big] \subseteq 2^X$}): A subset of links ${S}$ are permitted by Cumulative-interference-Power Carrier Sensing (CPCS) $(i.e., {S} \in {\mathscr C}^{\sf simp}_{\sf P, N}\big[{X},{\sf t}_{\rm cs}\big])$, if and only~if for all $i \in {S}$, 
\begin{equation}
{\sf N} + \underset{j \in {S} \backslash \{ i \}}{\sum} {\sf P}
|t_j - t_i|^{-\alpha} \le {\sf t}_{\rm cs} 
\end{equation}
where ${\sf t}_{\rm cs}$ is a parameter of carrier sensing threshold.
\end{definition}
In conventional CSMA devices, ${\sf t}_{\rm cs}$ is ${\sf c} + {\sf N}$, where ${\sf c}$ is a constant and does not depend on ${\sf N}, {\sf P}, \beta, \alpha$. 

\medskip

To ensure interference-safe transmissions under static CPCS, we need to find a universal carrier sensing threshold ${\sf t}_{\rm cs}$ such that
${\mathscr C}^{\sf simp}_{\sf P, N}\big[{X},{\sf t}_{\rm cs}\big] \subseteq {\mathscr B}^{}_{\sf P, N}\big[{X},\beta\big]$ for any $X$.
However, ${\mathscr C}^{\sf simp}_{\sf P, N}\big[{X},{\sf t}_{\rm cs}\big]$ does not take into account the ordering of local measurements, which can cause under-estimation of interference by the transmitters that started to transmit earlier.
We next present an example to illustrate this phenomenon.

\medskip

\begin{figure}[htb!] 
\centering \includegraphics[scale=0.75]{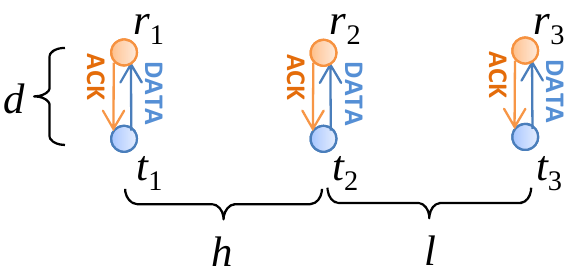}  \caption{\small Three transmitter-receiver pairs arranged in parallel.} \label{fig:ex1} \end{figure} 

\begin{example} ({\em Effect of Ordering}) We consider the scenario in Fig.~\ref{fig:ex1} with a given carrier sensing threshold ${\sf {\sf t}_{\rm cs}}$. The {\em locally measured} interference-and-noise power from the set of existing concurrent transmitters $X$ and the background noise measured at transmitter $x$ is denoted by: \begin{equation} 
{\bf P}_{x}[X] \triangleq {\sf N} + \sum_{z \in X}{\sf P} |z - x|^{-\alpha} 
\end{equation} 
With CPCS, transmitter $x$ proceeds with transmission, only if ${\bf P}_{x}[X] \le {\sf {\sf t}_{\rm cs}}$ for a pre-specified carrier sensing threshold ${\sf {\sf t}_{\rm cs}}$. 
In fact, whatever given value of ${\sf {\sf t}_{\rm cs}}$ for Fig.~\ref{fig:ex1}, there always exist values of $h, l$, such that the local interference power at a transmitter can exceed ${\sf {\sf t}_{\rm cs}}$ unawarely. In this example, assume ${\sf N}=0$ and ${\sf P}=1$. For a given ${\sf {\sf t}_{\rm cs}}$, we set $h, l$ such that they satisfy: 
\begin{equation} 
h^{-\alpha} = {\sf {\sf t}_{\rm cs}}, \quad l^{-\alpha} + (h+l)^{-\alpha} = {\sf {\sf t}_{\rm cs}} 
\end{equation} 
We suppose that the transmitters follows the sequence of local measurements as in Fig.~\ref{fig:ex2}. Initially there is no transmission. First, $t_1$ measures the interference power as ${\bf P}_{t_1}[\varnothing] = 0 \le {\sf {\sf t}_{\rm cs}}$, and proceeds with transmission. Next, $t_2$ measures the interference power as ${\bf P}_{t_2}[\{t_1\}] = {\sf {\sf t}_{\rm cs}}$, and proceeds with transmission. Lastly, $t_3$ measures the interference power as ${\bf P}_{t_3}[\{t_1, t_2\}] = {\sf {\sf t}_{\rm cs}}$, and also proceeds with transmission. However, after $t_3$'s transmission, the interference power at $t_2$ rises to ${\bf P}_{t_2}[\{t_1, t_3\}] > {\bf P}_{t_2}[\{t_1\}] = {\sf {\sf t}_{\rm cs}}$. This highlights the situation where an early transmitter ($t_2$) may be unaware that a later transmitter ($t_3$) can increase its interference power in the course of its transmission, even though the later transmitter ($t_3$) ensures its locally measured cumulative interference power to be below ${\sf {\sf t}_{\rm cs}}$. Therefore, it is insufficient to only consider ${\mathscr C}^{\sf simp}_{\sf P, N}\big[{X},{\sf t}_{\rm cs}\big]$.
\endexam \end{example}

\begin{figure}[htb!] 
\centering
\includegraphics[scale=0.68]{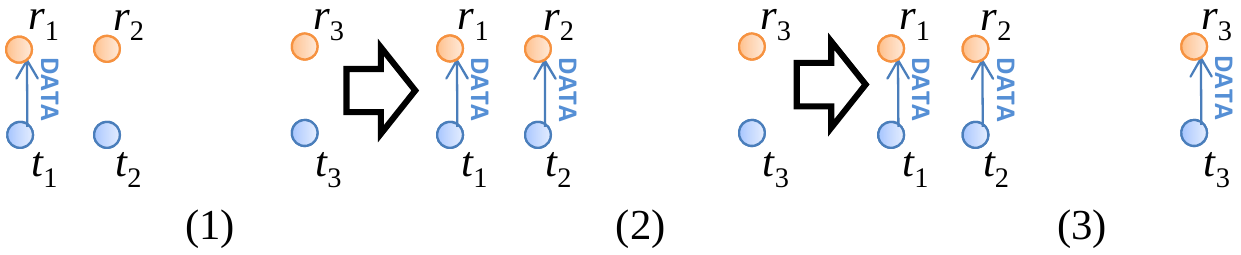}  \caption{\small A sequence of local interference power measurements.}  
\label{fig:ex2} \end{figure}

From the preceding example, it is insufficient to set a universal carrier sensing threshold ${\sf t}_{\rm cs}$ to obtain
${\mathscr C}^{\sf simp}_{\sf P, N}\big[{X},{\sf t}_{\rm cs}\big] \subseteq {\mathscr B}^{}_{\sf P, N}\big[{X},\beta\big]$ for any $X$.
We observe that the ordering effect always happens regardless of the value of carrier sensing threshold ${\sf {\sf t}_{\rm cs}}$. 
This raises the question that how we can set a robust carrier sensing threshold to ensure interference-safe transmissions despite of the ordering effect.  

As a remedy, we consider a feasible state in CPCS with local interference power measurement as a {\em sequence} of transmitters, instead of a subset of concurrent links.

\medskip

\begin{definition}
({\em The set of cumulative-interference-power carrier sensing feasible states ${\mathscr C}^{\sf cpcs}_{\sf P, N}\big[{X},{\sf t}_{\rm cs}\big]$}): We write ${\cal S}$ as a sequence $(i_1, ..., i_{|{\cal S}|})$, where each $i_k \in {X}$. A sequence of transmitters ${\cal S}$ is permitted by  Cumulative-interference-Power Carrier Sensing (CPCS)
$(i.e., {\cal S} \in {\mathscr C}^{\sf cpcs}_{\sf P, N}\big[{X},{\sf t}_{\rm cs}\big])$, if and only if for all~$i_k \in {\cal S}$, 
\begin{equation} \label{eqn:agcs_cond}
{\sf N} + \underset{i_j \in \{i_1, ..., i_{k-1}\}}{\sum} {\sf P}
|t_{i_j} - t_{i_k}|^{-\alpha} \le {\sf t}_{\rm cs} 
\end{equation}
where ${\sf t}_{\rm cs}$ is a parameter of carrier sensing threshold.
\end{definition}

\medskip

That is, when each $i_k$ observes the interference power from other concurrent transmitters that have already started transmissions before is below the carrier sensing threshold ${\sf t}_{\rm cs}$, $i_k$ decides that it is allowed to transmit.
By a slight abuse of notation, we also denote a sequence ${\cal S}$ as the set of its ordered items. For example, if ${\cal S} = (i_1, i_2, i_3)$,  ${\cal S}$ is also denoted as a set: ${\cal S} = \{i_1, i_2, i_3\}$.

Our goal is to study how to configure the carrier sensing threshold ${\sf t}_{\rm cs}$ to ensure interference-safe transmissions. Namely, we aim at finding a robust value of ${\sf t}_{\rm cs}$, such that 
\begin{equation}
{\cal S} \in {\mathscr C}^{\sf cpcs}_{\sf P, N}\big[{X},{\sf t}_{\rm cs}\big] \ \ \Rightarrow \ \ {\cal S} \in {\mathscr B}^{}_{\sf P, N}\big[{X},\beta\big] 
\end{equation}

We aim at finding the maximal value of ${\sf t}_{\rm cs}$ without the complete knowledge of ${X}$. In the following, we suppose that we only know the maximum transmission distance $d_{\max} \triangleq \max_{i \in X} |t_i - r_i|$ in a priori, when setting ${\sf t}_{\rm cs}$. Alternatively, $d_{\max}$ can be interpreted as the maximum tolerable transmission distance for a given ${\sf t}_{\rm cs}$. Note that even without the knowledge of $d_{\max}$, given fixed constants $\alpha, \beta, {\sf N}, {\sf P}$, the upper bound for transmission distance is $|t_i - r_i| \le (\frac{\sf P}{\beta {\sf N}})^{1/\alpha}$.

\subsection{IPCS}

An alternate approach to ensure interference-safe transmissions is to use the {\em pairwise} interference power as in Incremental-Power Carrier-Sensing (IPCS) \cite{FLH10IPCS}. We compare the two approaches in this section. We provide a sufficient condition to eliminate hidden nodes for IPCS, parallel to that of CPCS. Note that although the sufficient condition provided in \cite{FLH10IPCS} only considers the case with zero background noise (i.e., ${\sf N}=0$),  it is straightforward to generalize to the case with non-zero background noise.

The basic idea of IPCS is that $t_i$ can possibly estimate the distance to each individual concurrent transmitter $t_k$ by measuring the change of interference power. Suppose that initially $t_i$ measures the cumulative interference power as:
${\bf P}_{t_i}[X \backslash \{t_k\}]$. Then when $t_k$ transmits, the measured change of interference power at $t_i$ becomes: 
\begin{equation}
\Delta{\sf P}_i = {\bf P}_{t_i}[X] - {\bf P}_{t_i}[X \backslash \{t_k\}] = {\sf P} |t_k - t_i|^{-\alpha}, 
\end{equation}
which reveals the distance $|t_k - t_i|$. Suppose that each transmitter $t_i$ maintains a counter ${\sf cnt}_i$ (initially set as 0). When $t_i$ detects any change $\Delta{\sf P}_i$,
\begin{itemize}
\item if $\Delta{\sf P}_i \ge {\sf P} \cdot {\sf r}_{\rm cs}^{-\alpha}$,
then ${\sf cnt}_i \leftarrow {\sf cnt}_i + 1$.
\item if $\Delta{\sf P}_i \le -{\sf P} \cdot {\sf r}_{\rm cs}^{-\alpha}$,
then ${\sf cnt}_i \leftarrow {\sf cnt}_i - 1$.
\end{itemize}
where ${\sf r}_{\rm cs}$ is a parameter of separation between simultaneous transmitters.  
Transmitter $t_i$ is allowed to transmit only if ${\sf cnt}_i = 0$. 

In an idealized CSMA protocol \cite{BOE10}, congestion avoidance count-down is based on a continuous random variable. Probabilistically, no transmitters will simultaneously start to transmit at the same time by generating exactly the same value of countdown. In such a setting, IPCS realizes a set of carrier sensing feasible states as follows:

\medskip

\begin{definition}
({\em The set of incremental-interference-power carrier sensing feasible states ${\mathscr C}^{\sf ipcs}_{\sf P, N}\big[{X},{\sf r}_{\rm cs}\big] \subseteq 2^X$}): A subset of links are permitted by Incremental-interference-Power Carrier Sensing (CPCS)
$(i.e., {S} \in {\mathscr C}^{\sf ipcs}_{\sf P, N}\big[{X},{\sf r}_{\rm cs}\big])$, if and only~if for all $i, j \in {S}$, 
\begin{equation}
|t_j - t_i| \ge {\sf r}_{\rm cs} 
\end{equation}
where ${\sf r}_{\rm cs}$ is a parameter of separation between simultaneous transmitters.  
\end{definition}

\section{Static Carrier Sensing Mechanisms}

\subsection{CPCS} \label{sec:agg}

This section presents a proven approach to ensure interference-safe transmissions in CPCS by configuring a robust carrier sensing threshold ${\sf t}_{\rm cs}$. 

First, we rely on the notion of {\em interference level} at transmitter $t_i$ with respect to a subset of links ${S}$, which is defined as: 
\begin{equation}
{\bf I}_{t_i}[S, \alpha] \triangleq  \underset{j \in {S}}{\sum} |t_j - t_i|^{-\alpha} 
\end{equation}
We denote the maximal interference level in Euclidean space $\Re^{\sf d}$, subject to ${\mathscr C}^{\sf cpcs}_{\sf P, N}[{X},{\sf t}_{\rm cs}]$ with background noise ${\sf N}$, by: 
\begin{equation}
{\bf I}_{\max, \sf d}^{\sf cpcs}[{\sf t}_{\rm cs}, \alpha, {\sf P}, {\sf N}] \triangleq \max_{X, {\cal S} \in {\mathscr C}^{\sf cpcs}_{\sf P, N}[{X},{\sf t}_{\rm cs}], i \in {\cal S}} {\bf I}_{t_i}[{S}, \alpha]
\end{equation}

With a slight abuse of notation, we write the normalized maximal interference level as:
\begin{equation}
{\bf I}_{\max, \sf d}^{\sf cpcs}[\alpha] \triangleq {\bf I}_{\max, \sf d}^{\sf cpcs}[1, \alpha, 1, 0] 
\end{equation}
Note that ${\bf I}_{\max, \sf d}^{\sf cpcs}[\alpha]$ is a fundamental parameter depending on the dimension of the space ${\sf d}$, and provides a key theoretical tool to determine a robust carrier sensing threshold to ensure interference-safe transmissions.

\medskip

\begin{theorem} \label{thm:ag_imax}
Given a set of links $X$, which lies in Euclidean space $\Re^{\sf d}$, let $d_{\max} \triangleq \max_{i \in X} |t_i - r_i|$.
If we set 
\begin{equation}
{\sf t}_{\rm cs} \le  {\sf P} \Big( 2 d_{\max} +   
\big( \frac{1}{{\bf I}_{\max, \sf d}^{\sf cpcs}[\alpha]}
(\frac{d_{\max}^{-\alpha}}{\beta} - \frac{\sf N}{\sf P}) \big)^{\frac{-1}{\alpha}}
\Big)^{-\alpha}  + {\sf N} 
\end{equation}
then it can ensure interference-safe transmissions in CPCS, namely,
\begin{equation}
{\cal S} \in {\mathscr C}^{\sf cpcs}_{\sf P, N}\big[{X},{\sf t}_{\rm cs}\big] \ \ \Rightarrow \ \ {\cal S} \in {\mathscr B}^{}_{\sf P, N}\big[{X},\beta\big] 
\end{equation}
\end{theorem}
\begin{proof}
See the Appendix.
\end{proof}

\medskip

Note that $\frac{d_{\max}^{-\alpha}}{\beta} \ge \frac{\sf N}{\sf P}$. Otherwise, $\frac{{\sf P} d_{\max}^{-\alpha}}{\sf N} < \beta$, and the links separating by distance $d_{\max}$ cannot attain successful transmission, even without interferers.

Theorem~\ref{thm:ag_imax} provides a robust value of ${\sf t}_{\rm cs}$ to guarantee interference-safe transmissions regardless of the network topology. Although the value appears to be dependent on the parameters of channel model $\alpha$, our results are relatively conservative, and in the evaluation of Sec.~\ref{sec:eval}, we observe that the threshold is relatively robust in spite of uncertain parameters from the channel model. 

We note that when ${\sf N} = 0$ in particular, we obtain 
\begin{equation}
{\sf t}_{\rm cs}  \le  {\sf P}  \Big( \big( 2 + (\beta {\bf I}_{\max, \sf d}^{\sf cpcs}[\alpha])^{\frac{1}{\alpha}} \big) d_{\max} \Big)^{-\alpha} 
\end{equation}

In the appendix, we provide numerical analysis for ${\bf I}_{\max, \sf d}^{\sf cpcs}[\alpha]$ for 1-D and 2-D cases.
We show that ${\bf I}_{\max, \sf 2}^{\sf cpcs}[4] \approx 7.2$ when $\alpha = 4$ in 2-D case.

\subsection{IPCS} \label{sec:ipcs}
 
Similar to CPCS, we define
\begin{equation}
{\bf I}_{\max, \sf d}^{\sf ipcs}[{\sf r}_{\rm cs}, \alpha, {\sf P}, {\sf N}] \triangleq \max_{X, {\cal S} \in {\mathscr C}^{\sf ipcs}_{\sf P, N}[{X},{\sf r}_{\rm cs}], i \in {\cal S}} {\bf I}_{t_i}[{S}, \alpha]
\end{equation} 
and let the normalized maximal interference level be 
\begin{equation}
{\bf I}_{\max, \sf d}^{\sf ipcs}[\alpha] \triangleq {\bf I}_{\max, \sf d}^{\sf ipcs}[1, \alpha, 1, 0] 
\end{equation}
 
\begin{theorem} \label{thm:pw_imax}
Given a set of links $X$, which lies in Euclidean space $\Re^{\sf d}$, let $d_{\max} \triangleq \max_{i \in X} |t_i - r_i|$.
If we set 
\begin{equation}
{\sf r}_{\rm cs} \ge 
\bigg(
\frac{1}{{\sf P} {\bf I}_{\max, \sf d}^{\sf ipcs}[\alpha] } \Big( \frac{1}{ \beta}{\sf P} d_{\max}^{-\alpha} - {\sf N} \Big) \bigg)^{-\frac{1}{\alpha}} + 2 d_{\max} 
\end{equation}
then it can ensure interference-safe transmissions in IPCS, namely,
\begin{equation}
{\cal S} \in {\mathscr C}^{\sf ipcs}_{\sf P, N}\big[{X},{\sf r}_{\rm cs}\big] \ \ \Rightarrow \ \ {\cal S} \in {\mathscr B}^{}_{\sf P, N}\big[{X},\beta\big] 
\end{equation}
\end{theorem} 
\begin{proof}
The proof is similar to the one of Theorem~\ref{thm:ag_imax}. However, there is an alternate proof considering the case with zero background noise (i.e., ${\sf N}=0$) in \cite{FLH10IPCS},  which can be generalized to the case with non-zero background noise.
\end{proof}

\section{Adaptive Carrier Sensing Mechanisms}

In this section, we propose adaptive mechanisms to adjust the default parameters dynamically based on the feedback of nearby transmissions.

The basic idea is that, rather than as a universal value for every node, the carrier sensing parameter should be adapted differentially at each node according to its local network topology. For example, when the nodes are sparsely located, the carrier sensing parameter should be more lenient to allow more aggressive transmissions in typical settings. 

The principle for adjusting carrier sensing parameter boils down to balancing the local level of hidden nodes (HN) and exposed nodes (EN) in the network. We propose adaptive mechanisms that dynamically adjust the carrier sensing threshold ${\sf t}_{\rm cs}$ and ${\sf r}_{\rm cs}$ for CPCS and IPCS respectively with respect to the detected level of hidden nodes (HN) and exposed nodes (EN) at each node.

Let ${\sf t}^\ast_{\rm cs} $ and ${\sf r}^\ast_{\rm cs}$ be the default values of carrier sensing parameters for the static carrier sensing mechanisms, which can be obtained in Theorems~\ref{thm:ag_imax} and \ref{thm:pw_imax}. We first describe adaptive CPCS. We assume that the transmission power of nodes is initialized to be the maximum of the minimum required transmission powers for all links to remain connected in the network, and the initial carrier sensing threshold ${\sf t}^\ast_{\rm cs} $ is adjusted based on such a power level. 

For any node $j$ in the network, it keeps detecting the occurrence of two events throughout the process and adjusts the carrier sensing threshold as follows:
\begin{enumerate}

\item
If the node $j$ finds that it has not been allowed to transmit for $n_{\rm slot}$ or more consecutive packet slots, it is probable that it is suffering from the EN problem, and hence, it will attempt to increase its carrier sensing threshold ${\sf t}_{\rm cs} $ by $\delta_{s}$ to avoid the EN as long as it does not exceed the maximum carrier sensing threshold ${\sf t}_{\max}$. 

\item
On the other hand, if the node $j$ fails to receive ACK for $m_{\rm ack}$ or more consecutive transmitted packets, it could be suffering from the HN problem, and hence, it will broadcast a ``HN warning" signal through the beacon message to all $h_{\rm W}$-hop neighbors. 

\item
For every $h_{\rm W}$-hop neighbor that receives the HN warning signal, it will decrease its carrier sensing threshold ${\sf t}_{\rm cs} $ by $\delta_{s}$ as long as it is above the initial threshold ${\sf t}^\ast_{\rm cs} $. Otherwise, ${\sf t}_{\rm cs} $ remains unchanged. 

\end{enumerate}

The HN warning signal can be implemented by some simple modifications in the IEEE 802.11 MAC frame. Specifically, two fields need to be added or defined in the MAC header. First, a time-to-live (TTL) field \cite{TTL} is added to keep track of the $h_{\rm W}$ hops traversed by the HN warning signal (TTL field is available for IP datagram; here we are emulating that field at the MAC layer). Second, a broadcast mode for the HN warning signal is indicated in the MAC frame. Each time a HN warning signal is received, the TTL field will be decremented by 1 and the warning signal will be forwarded (further broadcasted) only if TTL is still larger than 0. Finally, through jointly considering the sequence number and source MAC address of the warning packet, a warning signal can be uniquely identified, so that duplicated warnings (e.g., if there are broadcast loop in the network) can be identified and dropped to avoid unnecessary carrier-sensing threshold adjustment.

The pseudo-code of adaptive CPCS is presented in Algorithm \ref{pseudocode}, and Table~\ref{para} summarizes the notations.

\begin{algorithm}
\caption{Adaptive CPCS}
\begin{algorithmic}[1] 
\State Initialize ${\sf t}_{\rm cs}={\sf t}^\ast_{\rm cs};$   
\For {$t = 1:T$}
\If{node $j$ fails to receive ACK for $m_{\rm ack}$\\ \qquad or more consecutive packets}
\LeftComment{{\rm\em Possible HN problem}}
\State Node $j$ broadcasts ``HN warning" to all \\ \qquad \quad $h_{\rm W}$-hop neighbors;
\EndIf
\If {node $j$ was not allowed to transmit for $n_{\rm slot}$\\ \qquad or more consecutive packets}
\LeftComment{{\rm\em  Possible EN problem}}
\If {${\sf t}_{\rm cs} <= {\sf t}_{\max}-\delta_{\rm s}$}
\State ${\sf t}_{\rm cs} \leftarrow {\sf t}_{\rm cs}+\delta_{\rm s};$
\EndIf
\ElsIf {node $j$ receives any ``HN warnings"\\ \qquad from its neighbors}
\LeftComment{{\rm\em It could be a HN to its neighbors}}
\If {${\sf t}_{\rm cs} \ge {\sf t}^\ast_{\rm cs}+\delta_{\rm s}$}
\State ${\sf t}_{\rm cs} \leftarrow {\sf t}_{\rm cs}-\delta_{\rm s};$
\EndIf
\Else
\State ${\sf t}_{\rm cs}$ remains unchanged;
\EndIf 
\EndFor
\end{algorithmic}
\label{pseudocode}
\end{algorithm}

The sensitivity of this adaptive algorithm towards HN and EN can be adjusted via the parameters $m_{\rm ack}$, $n_{\rm slot}$, $h_{\rm W}$, and $\delta_{s}$. For instance, a larger $n_{\rm slot}$ will reduce the false alarms for EN, while a larger $m_{\rm ack}$ and smaller $h_{\rm W}$ will reduce the false alarms for HN. The step size $\delta_{s}$ for adjusting the carrier sensing threshold will also affect the sensitivity of the algorithm and hence the performance. 

In addition, the above carrier sensing threshold adaptation process can also be applied to IPCS to configure the carrier sensing range ${\sf r}_{\rm cs} $ accordingly, such that the maximum interference levels at a transmitter by CPCS and IPCS reach the same level. For adaptive IPCS, we map ${\sf t}_{\rm cs}$ in Algorithm \ref{pseudocode} to ${\sf r}_{\rm cs}$ by the following equation:
\begin{equation} \label{eqn:TCStoRCS}
{\sf t}_{\rm cs} =  {\sf P} ({\sf r}_{\rm cs} )^{-\alpha} + {\sf N} \\
\quad \Rightarrow \quad {\sf r}_{\rm cs} =  \big( \frac{{\sf t}_{\rm cs} - {\sf N}}{\sf P}   \big)^\frac{-1}{\alpha}
\end{equation}
The above equation relates the receiver's interference power level (${\sf t}_{\rm cs}$) from a distance ${\sf r}_{\rm cs}$, when there is one interfering transmitter.

\begin{table}
	\centering
	\caption{Table of Parameters.}
    \begin{tabular}{  c | p{5cm} }
    \hline\hline
    Parameter & Description  \\ \hline
    \ $\delta_{\rm s}$ & The step size for the carrier sensing threshold adaptation  \\ \hline
    \ ${\sf t}^\ast_{\rm cs} $ & The initial carrier sensing threshold identified by static CPCS  \\ \hline
    \ ${\sf t}_{\max}$ & The maximum allowed carrier sensing threshold \\ \hline
    \ $m_{\rm ack}$ & The number of consecutive transmissions with no ACK received\\ \hline
    \ $n_{\rm slot}$ & The number of consecutive packet slots prevent to transmit\\ \hline
    \ $h_{\rm w}$ & The number of hops for propagation of ``HN warning" \\ \hline\hline
    \end{tabular}
    \label{para}
\end{table}

\section{Performance Evaluation} \label{sec:eval}

We implemented static and adaptive versions of CPCS and IPCS in a simulation test-bed in MATLAB considering infrastructure-less WiFi networks. We consider the IEEE 802.11 standard, and MAC layer protocols such as CSMA/CA and distributed coordination function (DCF) are implemented. The common physical layer data rate is 11 Mbps. The slot time is 20$\mu$s, and the SIFS and DIFS are 10$\mu$s and 50$\mu$s respectively. The SIR requirement is 20dB. Table~\ref{parasetting} summarizes the settings of the simulation.

Specifically, we look at the aggregated goodput, packet failure rate, and Jain's fairness index defined by
\begin{equation}
\mbox{Jain's fairness index} = \frac{(\sum_{i=1}^n \lambda_i)^2}{n\sum_{i=1}^n \lambda_i}
\end{equation}
where $\lambda_i$ is the goodput of link $i$ in the total $n$ links. A larger Jain's fairness index implies the fairer distribution of goodput. We also study the phenomenon of starvation, namely the portion of links with low goodput, in adaptive CPCS and IPCS. We define node density to be the average number of nodes per transmission range. We simulated up to 1000 instances with respect to different node density in each simulation.

\subsection{Static CPCS and IPCS} 

In this section, we compare the performance of static CPCS and IPCS. We consider settings with 50 to 300 links in a 3000m$\times$3000m area. The length of a link follows uniform distribution between 10m and 250m. We consider two types of network topologies: (1) random networks where the coordinates of nodes are uniformly distributed in the space, and (2) clustered networks where the links are randomly distributed around certain clusters in the space. Clustered networks are more skew in node density.

We set the carrier sensing parameters in CPCS and IPCS to prevent hidden nodes, according to Theorems~\ref{thm:ag_imax} and \ref{thm:pw_imax}. Because of the absence of hidden nodes, there is zero packet failure rate. We evaluated the aggregate goodput and Jain's fairness index.

For $\alpha = 4$, we plot the performance against the number of links in Fig.~\ref{fig:APvsIP}. For aggregate goodput, we observe that CPCS performs comparably to IPCS in the uniformly random networks, but outperforms IPCS in the clustered networks. In general, both CPCS and IPCS perform significantly better in the clustered networks, because of denser node distribution in the clustered networks. For fairness, CPCS outperforms IPCS significantly, because CPCS depends on the cumulative interference in the neighborhood, rather than pairwise interference that can result in unbalanced transmissions.

In comparison with traditional 802.11 CSMA (which sets the carrier sensing threshold as 20 dB above the background noise level), CPCS is more conservative in order to prevent hidden node problem. We observe that traditional 802.11 CSMA can cause a high packet failure rate when the background noise is large.

\begin{figure}[htb!] \vspace{-0pt}
\hspace{-10pt}
    \includegraphics[scale=0.45]{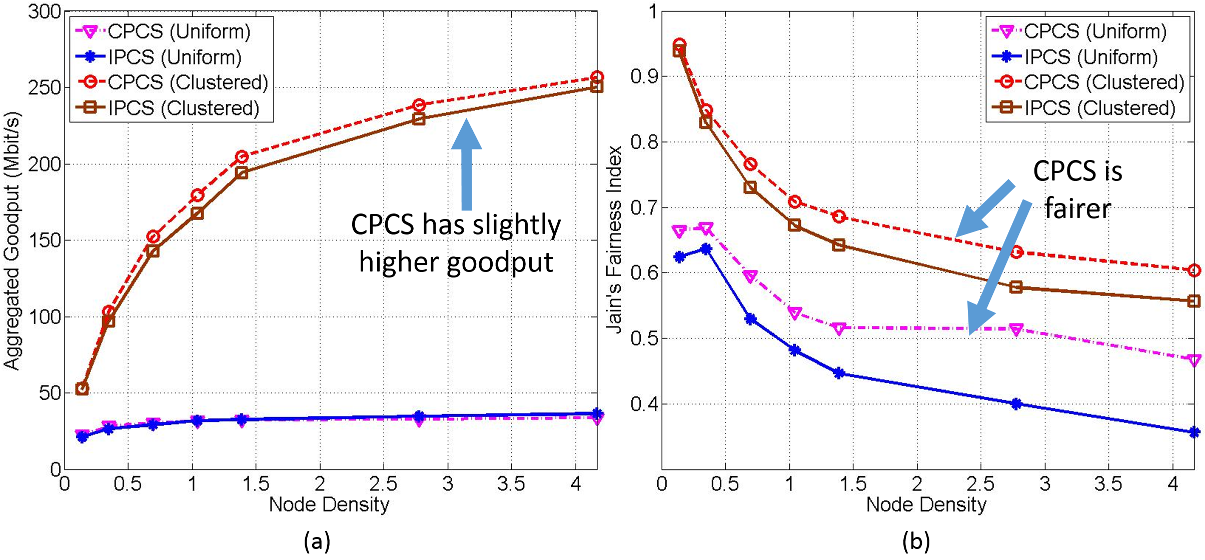}
     \vspace{-5pt}
    \caption{For $\alpha = 4$, (a) aggregated goodput and (b) Jain's fairness index for CPCS and IPCS for uniformly random networks and clustered networks against the node density.} \label{fig:APvsIP} 
\end{figure}

\begin{figure*}[htb!] 
\centering
\includegraphics[scale=0.55]{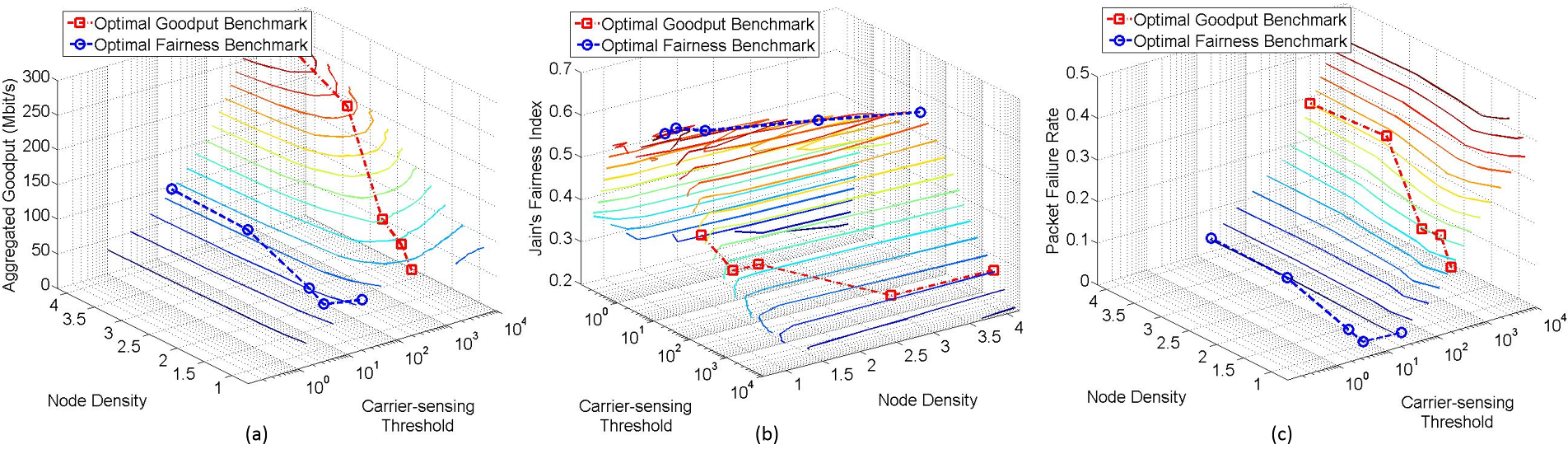} \vspace{-5pt}
    \caption{3D contour plots of (a) aggregated goodput; (b) Jain's fairness index; and (c) packet failure rate ($\alpha = 4$, $K_a = \infty$, $\delta_{\rm s}/{\sf t}^\ast_{\rm cs}  = 20$), with respect to specific carrier sensing threshold ${\sf t}_{\rm cs}$ and node density $n$.} \label{fig:3Dplot} 
\end{figure*}

\subsection{Adaptive CPCS and IPCS}

In this section, we compare the performance of adaptive CPCS and IPCS. 
The maximum carrier sensing threshold ${\sf t}_{\max}$ and the adjustment step size $\delta_{\rm s}$ control the scale and rate of change of the adaptation process, respectively. In our simulation, we set ${\sf t}_{\max}/{\sf t}^\ast_{\rm cs}=10^4$ and $\delta_S/{\sf t}^\ast_{\rm cs} =20$, because we observe that such a range and step size allow us to identify the optimal performance (goodput and fairness) with respect to the carrier sensing threshold within a reasonable number of iterations.

\begin{table}
	\centering
	\caption{Simulation Parameters.}
    \begin{tabular}{ l | l || l | p{1.5 cm} }
    \hline\hline
    Parameter & Value & Parameter & Value \\ \hline
    \# links & 50 - 500 & $\alpha$ & 3, 4  \\ \hline
    Packet size & 1460 bytes & $CW_{\min}$ & 31  \\ \hline
     ${\sf t}_{\max}/{\sf t}^\ast_{\rm cs} $ & $10^4$ & $\delta_{\rm s}/{\sf t}^\ast_{\rm cs} $ & 20
     \\ \hline
  	$m_{\rm ack}$ & 1, 2 & $n_{\rm slot}$ & 3, 6  \\ \hline
  	$h_{\rm W}$ & 1-10 & $K_a$ & 0.1, 1, 10, $\infty$\\ \hline\hline
    \end{tabular}
    \label{parasetting}
\end{table}

\subsubsection{Benchmarks}

To obtain the insight of performance of adaptive CPCS and IPCS, we compare them with certain benchmarks in ideal situations.
We consider uniformly distributed random network topologies.

First, we evaluated the performance (e.g., aggregate goodput, Jain's fairness index, failure rate) under specific values of ${\sf t}_{\rm cs}$. Fig. \ref{fig:3Dplot} shows the 3D contour plots of how the performance metrics vary with respect to the node density and carrier sensing threshold ${\sf t}_{\rm cs}$, when we set a uniform threshold ${\sf t}_{\rm cs}$ in all the nodes. 

Next, we identify two benchmarks in Fig. \ref{fig:3Dplot}: (1) {\em optimal goodput benchmark}, such that ${\sf t}_{\rm cs}$ attains the highest aggregate goodput with the respect of node density, (2) {\em optimal fairness benchmark}, such that ${\sf t}_{\rm cs}$ attains the highest Jain's fairness index with the respect of node density. These benchmarks represent the ideal settings of carrier sensing parameters when the network topology is uniformly random. We note that the sandwiched region (as depicted in Fig. \ref{fig:3Dplot2}) between two benchmarks represents the operational regime that adaptive CPCS and IPCS should operate. We remark that, unlike the benchmarks using a uniform carrier sensing threshold ${\sf t}_{\rm cs}$, adaptive CPCS and IPCS allow heterogeneous carrier sensing thresholds at different nodes, and hence, can attain better performance in certain cases.

\begin{figure}[htb!]   
\centering
    \includegraphics[scale=0.55]{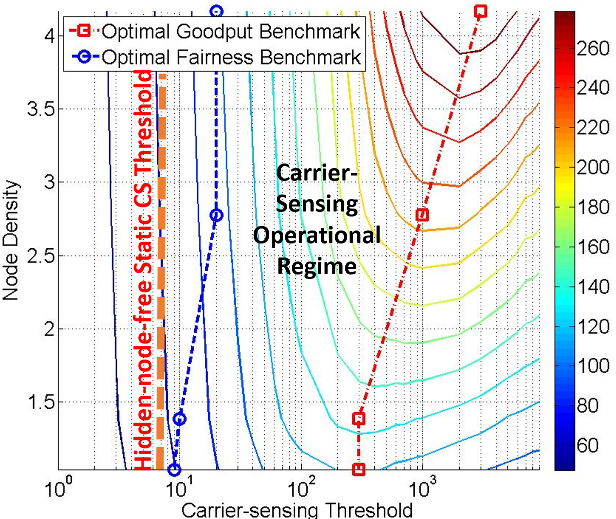}
	\vspace{-5pt}
    \caption{The top view of Fig.~\ref{fig:3Dplot} (a), with the optimal goodput benchmark and  optimal fairness benchmark. The plot also shows the hidden-node-free carrier sensing threshold for static CPCS.} \label{fig:3Dplot2} 
\end{figure}

\begin{figure*}[htb!] 
\centering
\includegraphics[scale=0.55]{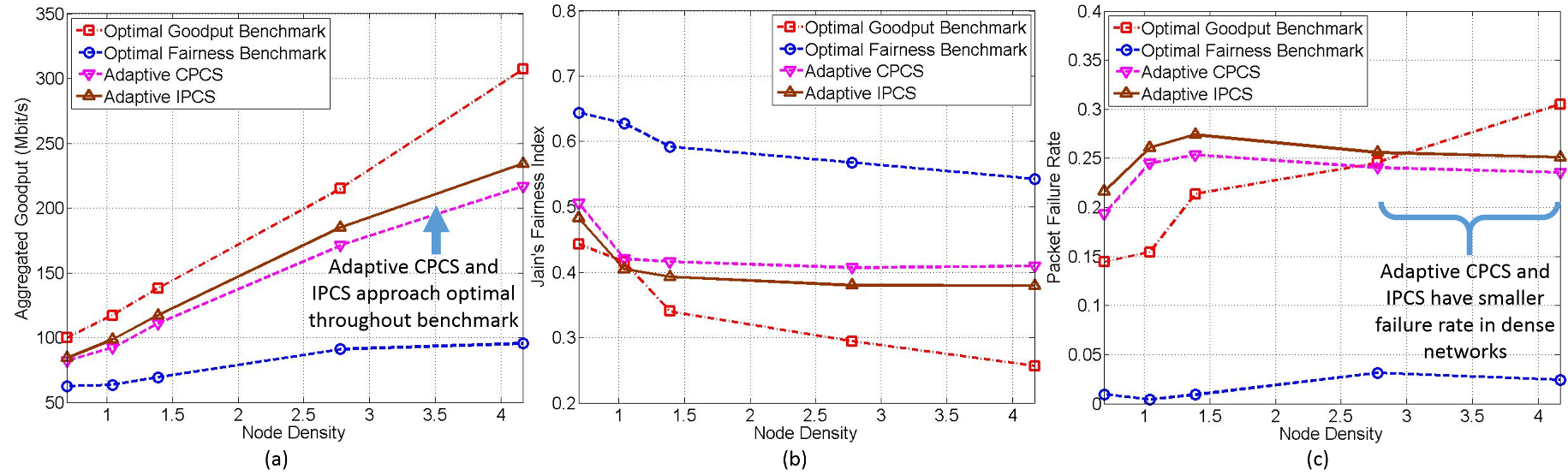} \vspace{-5pt}
    \caption{Plots of (a) aggregated goodput; (b) Jain's fairness index; and (c) packet failure rate for the optimal goodput benchmark, the optimal fairness benchmark, Adaptive CPCS, and Adaptive IPCS ($\alpha = 4$, $K_a =  \infty$, $\delta_{\rm s}/{\sf t}^\ast_{\rm cs}  = 20$, $m_{\rm ack}=2$, $h_{\rm W}=1$).} \label{fig:2Dsideplot2} 
\end{figure*}

\begin{figure*}[htb!]  
\centering
\includegraphics[scale=0.55]{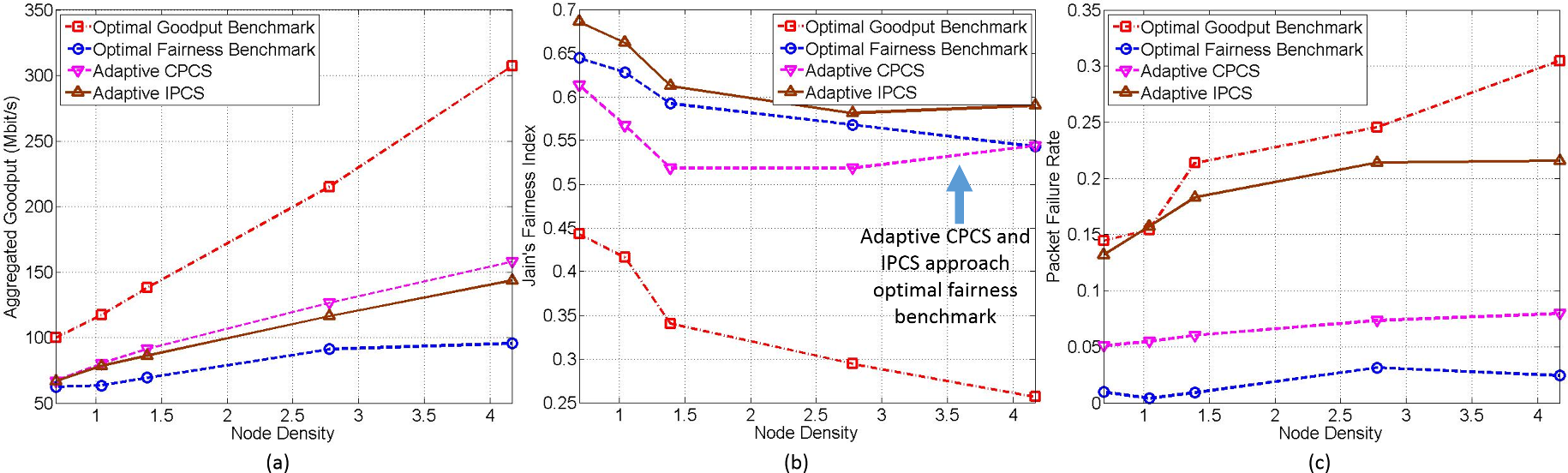} \vspace{-5pt}
    \caption{Plots of (a) aggregated goodput; (b) Jain's fairness index;  and (c) packet failure rate for the optimal goodput benchmark, the optimal fairness benchmark, Adaptive CPCS, and Adaptive IPCS ($\alpha = 4$, $K_a =  \infty$, $\delta_{\rm s}/{\sf t}^\ast_{\rm cs}  = 20$,  $m_{\rm ack}=1$, $h_{\rm W}=10$).} \label{fig:2Dsideplot}  
\end{figure*}

\subsubsection{Observations}

Figs. \ref{fig:2Dsideplot2}-\ref{fig:2Dsideplot} show the aggregated throughout, Jain's fairness index and packet failure rate for adaptive CPCS and IPCS, as compared with the optimal goodput benchmark and optimal fairness benchmark, for two settings of parameters: ($\alpha = 4$, $K_a =  \infty$, $\delta_{\rm s}/{\sf t}^\ast_{\rm cs}  = 20$,  $m_{\rm ack}=2$, $h_{\rm W}=1$) in Fig. \ref{fig:2Dsideplot2} and ($\alpha = 4$, $K_a =  \infty$, $\delta_{\rm s}/{\sf t}^\ast_{\rm cs}  = 20$, $m_{\rm ack}=1$, $h_{\rm W}=10$) in Fig. \ref{fig:2Dsideplot}.

In general, we observe that the adjustment triggering parameters in Algorithm~1, such as $m_{\rm ack}$ and $h_{\rm W}$, can be used to characterize the objective of the adaptation process. For instance, a smaller value of $m_{\rm ack}$ and a larger value of $h_{\rm W}$ can generate more HN warnings, and hence, propagate them further, which suppresses HN warnings more aggressively and favours fairness instead of goodput.

In particular, we observe that the adaptive mechanisms can achieve up to about 80\% of the aggregate goodput benchmark (Fig. \ref{fig:2Dsideplot2} (a)), and have better performance in terms of fairness and failure rate (especially when the node density is high). When comparing the adaptive CPCS and IPCS, they have very close performance. Specifically, adaptive IPCS performs slightly better in terms of aggregated goodput at the expense of fairness and failure rate.

When $m_{\rm ack}=1$ and $h_{\rm W}=10$ in Fig. \ref{fig:2Dsideplot}, adaptive CPCS can achieve similar fairness as the fairness benchmark while 50\% better goodput in high density scenarios. In addition, we observe that through adjusting the triggers (e.g., $m_{\rm ack}$ and $h_{\rm W}$) in Algorithm \ref{pseudocode}, the adaptive mechanisms slightly favor fairness. 

\subsubsection{Starvation}

Starvation is an important phenomenon, in which some links receive persistently lower goodput than others, despite high aggregate goodput over all the links \cite{KK12csma}. A desirable CSMA protocol should attain balanced goodput in all the links.  We plot the histogram of link goodput distribution of 500 simulated links as illustrated in Fig.~\ref{fig:starvation} (a). We observe that 65\% of the links are starved in the optimal goodput benchmark, whereas only 42\% are starved in adaptive CPCS and IPCS. To compare starvation, we introduce the notion of link starvation ratio, defined as
 \begin{equation}\label{eqn:StarvationRatio}
S=\dfrac{\text{node density with goodput} \leq \rho}{\text{total node density\ } n}
  \end{equation}
 
Here, we set the threshold $\rho$  to be the median goodput of adaptive CPCS, which is used as a reference of comparison. In Fig.~\ref{fig:starvation} (b), we plot the starvation ratio of the optimal goodput benchmark, as compared with adaptive CPCS and IPCS. 
We observe that adaptive CPCS and IPCS are effective in reducing starvation. Especially in dense networks, the starvation ratio is moderate when the node density increases.

\begin{figure*}[htb!] 
\centering
\includegraphics[scale=0.55]{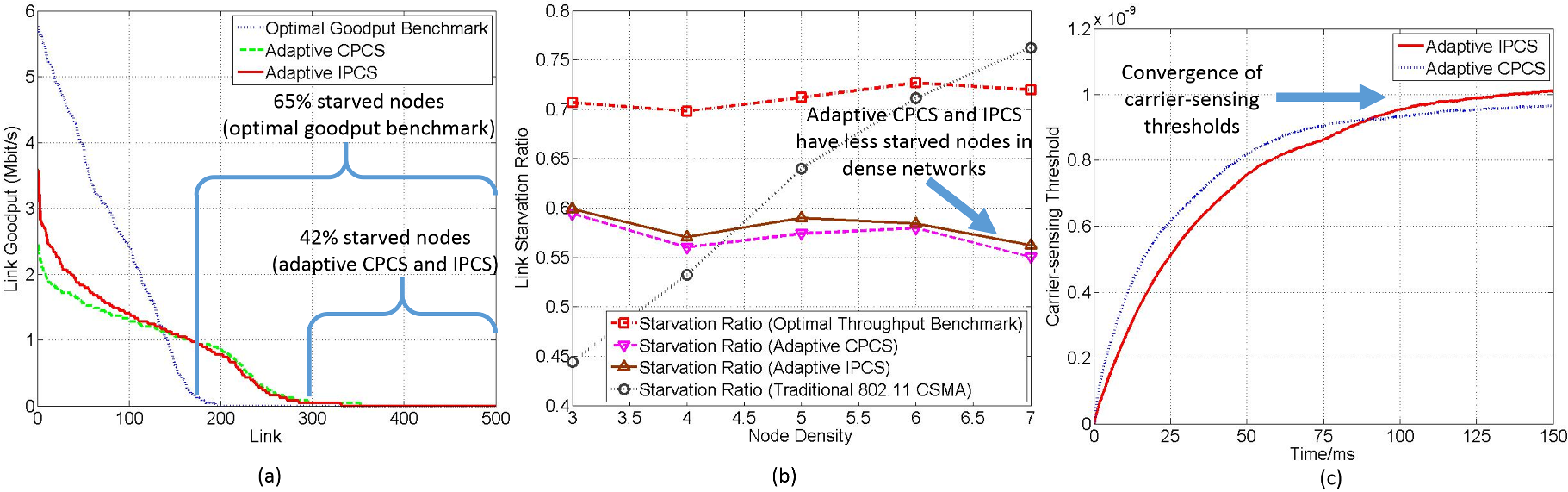} \vspace{-5pt}
    \caption{The goodput histogram for a scenario of 500 links that are (a) sorted by goodput; and (b) the corresponding starvation ratio up with respect to node density ($\rho$ is the median goodput of adaptive CPCS, $\alpha = 4$, $K_a =  \infty$, $\delta_{\rm s}/{\sf t}^\ast_{\rm cs}  = 20$). (c) shows the convergence behavior of adaptive CPCS and IPCS. } \label{fig:starvation} 
\end{figure*}

\subsubsection{Convergence}

Fig.~\ref{fig:starvation} (c) shows the convergence behavior of carrier sensing thresholds in adaptive CPCS and IPCS considering a scenario with 300 links. We observe that the convergence of thresholds converge to stable states rapidly after about 80 ms, which appears to be robust in the presence with moderate degree of mobility in the networks.

\begin{figure*}[htb!] 
\centering
\includegraphics[scale=0.55]{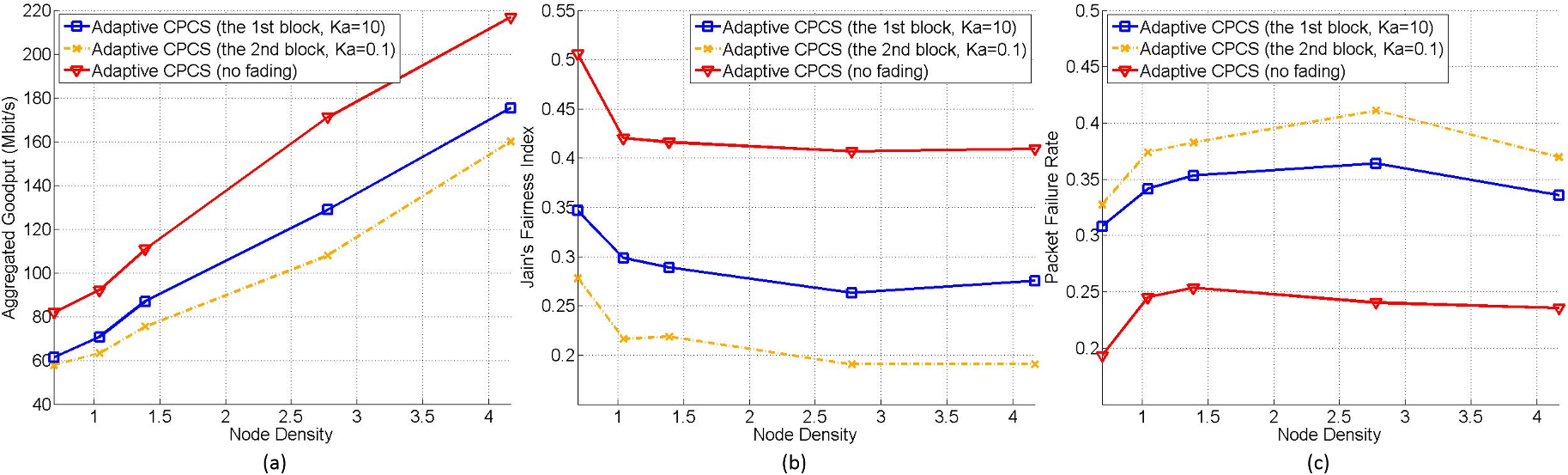} \vspace{-5pt}
    \caption{(a) Aggregated goodput; (b) Jain's fairness index; and (c) packet failure rate for optimal goodput benchmark and adaptive CPCS in the two-block fading scenario ($\alpha = 4$, first block: $K_a = 10$, second block: $K_a = 0.1$ $\delta_{\rm s}/{\sf t}^\ast_{\rm cs}  = 20$).} \label{fig:TwoBlock} 
\end{figure*}

\subsubsection{Dynamic Channel with Fading}

We also evaluated the robustness of the proposed adaptive mechanisms under a dynamic wireless communication channel with fading. To model a dynamic wireless communication channel with fading, we considered two random fading gains in two time blocks within the simulation period. The small-scale fading gain $\beta$ is described by the following distribution under Rician fading:
   \begin{multline}\label{eqn:Ricianfading} 
    p(\beta)=(1+K_a)\exp(-K_a)\exp(-(1+K_a)\beta)\cdot \\
        I_0(\sqrt{4K_a(1+K_a)\beta})
\end{multline}
We model the severity of fading in the channel through adjusting the Rician parameter $K_a$ that describes the ratio between the power for the line-of sight path and the power for other scattered paths. The larger the value of $K_a$, the less severe the fading is. That is, when $K_a=0$, it becomes Rayleigh fading; while $K_a=\infty$ represents no fading. $I_0(.)$ denotes the zeroth-order modified Bessel function of the first kind.  

Fig.~\ref{fig:TwoBlock} shows the performance of adaptive CPCS under such two-block fading channel with $K_a = 10$ in the first block and $K_a = 0.1$ (much closer to Rayleigh fading) in the second block. In the presence of fading, the performance of adaptive CPCS inevitably degrades, because the adaptation process will be obstructed by the variability of interference power measurement. 
In general, the performance of the adaptive mechanisms under such dynamic channel is better when the network becomes denser. After two blocks with severely different fading levels, adaptive CPCS can still achieve more than 70\% of the optimal goodput benchmark, and with fairness and failure rate close to that of the benchmark when the node density is high (more than 3).

\section{Discussions}

The implementation of the proposed carrier sensing threshold adjustment algorithms on practical systems can be achieved through the Clear Channel Assessment (CCA) in the IEEE 802.11 standard \cite{IEEE802112007}. In the standard, there are two modes of carrier sensing: 
\begin{enumerate}

\item[(i)] carrier sensing by energy detection (ED), and 

\item[(ii)] carrier sensing by frame detection and interpretation of the network allocation vector (NAV) in the MAC header therein \cite{JL08hidden}.

\end{enumerate}

For (i), the medium is declared to be busy as long as the detected power level is above a certain threshold, regardless of whether an 802.11 frame can be decoded or not. For (ii), an 802.11 MAC frame must first be detected and at least the header must be decodable. The NAV in a DATA frame typically specifies that the medium should be interpreted as busy for the duration of the DATA frame plus SIFS plus the duration of an ACK that follows (if RTS/CTS mode is turned up, then the durations of the RTS, CTS plus two more SIFS should also be added). 

IPCS and CPCS in this paper will only make use of carrier sensing by ED (i.e., (i)) \cite{FLH10IPCS}. Assuming RTS/CTS is not used, receivers that detect energy levels above the CS threshold can automatically add the durations of a SIFS and ACK as a safety measure. 

The adaptive mechanisms require the distribution of HN (hidden node) warning signal to adjacent nodes, which can be implemented by simple modifications in the IEEE 802.11 MAC frame. Specifically, two fields need to be added or defined in the MAC header, including a time-to-live (TTL) field to keep track of the number of hops traversed by the HN warning signal, and a bit to indicate a broadcast mode for the HN warning signal. Currently, certain MAC frame header fields are not commonly used (e.g., power management indication bit), which can be reused for secondary applications.

The parameters of the adaptive mechanisms and the network environment control the signaling overhead and convergence rate of the proposed algorithm. The signaling overhead (HN Warnings) in the algorithm depends on the node density (average number of nodes per transmission range) and $h_{\rm W}$ (number of hops the warnings propagate). When $h_{\rm W}=1$, we observe a linear relationship between the number of HN Warnings and the number of successful packet receptions, as the warnings only propagate one hop. The ratio is about 0.14 and 0.17 for adaptive CPCS and IPCS respectively. The size of HN warning is assumed to be small (e.g., $\approx$200 bits). Therefore, when one packet (with a size of $1460\times 8=11680$ bits) is successfully transmitted, the average signaling overhead would be no more than $0.17\times 200=34$ bits, which is rather moderate, as compared with the data traffic.


\section{Conclusion} \label{sec:concl}

To mitigate the interference among simultaneous transmitting wireless devices (e.g., for Internet-of-things), more effective CSMA protocols are proposed in this paper. The interactions between links in realistic CSMA networks are affected by the special properties attributed to SINR, effects of arbitrary ordering of local measurements, and ACK frames. This paper presents a viable standard-compatible solution to ensure interference-safe transmissions by determining a robust interference-safe carrier sensing threshold in Cumulative-interference-Power Carrier-Sensing (CPCS). Moreover, we present adaptive approaches that adjusts the carrier sensing thresholds dynamically based on the feedback of nearby transmissions. We provide extensive simulation studies to evaluate the starvation ratio, fairness, and goodput performance of our approaches.

\bibliographystyle{IEEEtran}
\bibliography{paperbib}

\section{Appendix}

\subsection{Proofs}

Define the {\em bi-directional} interference level at link $i$ with respect to a subset of links ${S}$ as: 
\begin{equation}
{\bf B}_{i}[S, \alpha] \triangleq  \underset{j \in {S}\backslash\{i\}}{\sum} {\sf dist}(i,j)^{-\alpha}  
\end{equation}
We denote the maximal bi-directional interference level in Euclidean space $\Re^{\sf d}$, subject to ${\mathscr C}^{\sf cpcs}_{\sf P, N}[{X},{\sf t}_{\rm cs}]$ with background noise ${\sf N}$, by: 
\begin{equation}
{\bf B}_{\max, \sf  d}^{\sf cpcs}[{\sf t}_{\rm cs}, \alpha, {\sf P}, {\sf N}] \triangleq \max_{X, {\cal S} \in {\mathscr C}^{\sf cpcs}_{\sf P, N}[{X},{\sf t}_{\rm cs}], i \in {\cal S}} {\bf B}_{i}[{S}, \alpha] 
\end{equation}

\medskip

\begin{lemma} \label{lem:bi_intf}
Let $d_{\max} \triangleq \max_{i \in X} |t_i - r_i|$. For a pair of ${\sf t}_{\rm cs}$ and ${\sf t}'_{\rm cs}$, if  
\begin{equation} \label{eqn:bi_intf_cond}
\big( \frac{|{\sf t}_{\rm cs} - {\sf N}|}{{\sf P}} \big)^{\frac{-1}{\alpha}} \ge \big( \frac{|{\sf t}'_{\rm cs} - {\sf N}|}{{\sf P}} \big)^{\frac{-1}{\alpha}} +  2 d_{\max}, 
\end{equation}
then 
\begin{equation}
{\bf B}_{\max, \sf  d}^{\sf cpcs}[{\sf t}_{\rm cs}, \alpha, {\sf P}, {\sf N}] \le {\bf I}_{\max, \sf d}^{\sf cpcs}[{\sf t}'_{\rm cs}, \alpha, {\sf P}, {\sf N}] 
\end{equation}
\end{lemma}

\begin{proof}
Suppose ${\cal S} = (i_1, ..., i_{|{\cal S}|}) \in {\mathscr C}^{\sf cpcs}_{\sf P, N}[{X},{\sf t}_{\rm cs}]$. By Eqn.~(\ref{eqn:agcs_cond}), for any pair $i_j, i_k \in {\cal S}$ and $j < k$, we obtain 
\begin{equation*}
{\sf N} + {\sf P} |t_{i_j} - t_{i_k}|^{-\alpha} \le {\sf t}_{\rm cs}
\end{equation*} 
\begin{equation*}
\Rightarrow\ 
|t_{i_j} - t_{i_k}| \ge \big( \frac{|{\sf t}_{\rm cs} - {\sf N}|}{{\sf P}} \big)^{\frac{-1}{\alpha}} \ge \big( \frac{|{\sf t}'_{\rm cs} - {\sf N}|}{{\sf P}} \big)^{\frac{-1}{\alpha}} +  2 d_{\max} 
\end{equation*}

Since $|t_{i_j} - r_{i_j}| \le d_{\max}$ and $|t_{i_k} - r_{i_k}| \le d_{\max}$, by triangular inequality, we obtain the following inequalities:
\begin{eqnarray*}
|t_{i_j} - r_{i_k}| \ge |t_{i_j} - t_{i_k}| - |t_{i_k} - r_{i_k}| \ge \big( \frac{|{\sf t}'_{\rm cs} - {\sf N}|}{{\sf P}} \big)^{\frac{-1}{\alpha}} + d_{\max} \\
|r_{i_j} - t_{i_k}| \ge |t_{i_j} - t_{i_k}| - |r_{i_j} - t_{i_j}| \ge \big( \frac{|{\sf t}'_{\rm cs} - {\sf N}|}{{\sf P}} \big)^{\frac{-1}{\alpha}} + d_{\max} \\
|r_{i_j} - r_{i_k}| \ge |r_{i_j} - t_{i_k}| - |t_{i_k} - r_{i_k}| \ge \big( \frac{|{\sf t}'_{\rm cs} - {\sf N}|}{{\sf P}} \big)^{\frac{-1}{\alpha}} 
\end{eqnarray*}
Hence, 
\begin{equation*}
{\sf dist}(i,j) \ge \big( \frac{|{\sf t}'_{\rm cs} - {\sf N}|}{{\sf P}} \big)^{\frac{-1}{\alpha}}
\end{equation*}
\begin{equation} \label{eqn:agcs_newcond}
\Rightarrow\  {\sf N} + \underset{i_j \in \{i_1, ..., i_{k-1}\}}{\sum} {\sf P} \cdot
{\sf dist}(i_j,i_k)^{-\alpha} \le {\sf t}'_{\rm cs}
\end{equation} 

We define a new set of feasible states $\tilde{\mathscr C}^{\sf cpcs}_{\sf P, N}\big[{X},{\sf t}'_{\rm cs}\big]$, such that
${\cal S} \in \tilde{\mathscr C}^{\sf cpcs}_{\sf P, N}\big[{X},{\sf t}'_{\rm cs}\big]$, if and only if Eqn.~(\ref{eqn:agcs_newcond}) holds for all $i_k \in {\cal S}$. It follows that
\begin{equation}
{\cal S} \in {\mathscr C}^{\sf cpcs}_{\sf P, N}\big[{X},{\sf t}_{\rm cs}\big]
\ \Rightarrow\  {\cal S} \in \tilde{\mathscr C}^{\sf cpcs}_{\sf P, N}\big[{X},{\sf t}'_{\rm cs}\big]
\end{equation}
Namely, ${\mathscr C}^{\sf cpcs}_{\sf P, N}\big[{X},{\sf t}_{\rm cs}\big] \subseteq \tilde{\mathscr C}^{\sf cpcs}_{\sf P, N}\big[{X},{\sf t}'_{\rm cs}\big]$. Hence,  
\begin{eqnarray*}
{\bf B}_{\max, \sf  d}^{\sf cpcs}[{\sf t}_{\rm cs}, \alpha, {\sf P}, {\sf N}] & = & \max_{X, {\cal S} \in {\mathscr C}^{\sf cpcs}_{\sf P, N}[{X},{\sf t}'_{\rm cs}], i \in {\cal S}} {\bf B}_{i}[{S}, \alpha]  \\
& \le & \max_{X, {\cal S} \in \tilde{\mathscr C}^{\sf cpcs}_{\sf P, N}[{X},{\sf t}'_{\rm cs}], i \in {\cal S}} {\bf B}_{i}[{S}, \alpha] 
\end{eqnarray*}
We complete the proof by: 
\begin{equation}
\max_{X, {\cal S} \in \tilde{\mathscr C}^{\sf cpcs}_{\sf P, N}[{X},{\sf t}'_{\rm cs}], i \in {\cal S}} {\bf B}_{i}[{S}, \alpha] \le {\bf I}_{\max, \sf d}^{\sf cpcs}[{\sf t}'_{\rm cs}, \alpha, {\sf P}, {\sf N}] 
\end{equation}
\end{proof}

\medskip

\begin{lemma} \label{lem:ag_rescale}
${\bf I}_{\max, \sf d}^{\sf cpcs}[{\sf t}_{\rm cs}, \alpha]$ has the following equivalence relations: 
\begin{enumerate}

\item {\em Rescaling carrier sensing threshold to ${\sf N} + 1$}: 
\begin{equation}
 {\bf I}_{\max, \sf d}^{\sf cpcs}[{\sf t}_{\rm cs}, \alpha, {\sf P}, {\sf N}] =  |{\sf t}_{\rm cs} - {\sf N}| \cdot {\bf I}_{\max, \sf d}^{\sf cpcs}[{\sf N} + 1, \alpha, {\sf P}, {\sf N}] \label{eqn:ag_rescale1}
\end{equation} 

\item {\em Rescaling transmission power to $1$}: 
 \begin{equation}
 {\bf I}_{\max, \sf d}^{\sf cpcs}[{\sf t}_{\rm cs}, \alpha, {\sf P}, {\sf N}] =  \frac{1}{\sf P} {\bf I}_{\max, \sf d}^{\sf cpcs}[{\sf t}_{\rm cs}, \alpha, 1, {\sf N}] \label{eqn:ag_rescale2}
\end{equation}

\end{enumerate}
\end{lemma}

\begin{proof}
The equivalence relations Eqns.~(\ref{eqn:ag_rescale1})-(\ref{eqn:ag_rescale2}) are shown by rescaling the interference level appropriately to preserve the same feasible state, considering different values of ${\sf t}_{\rm cs}$ and ${\sf P}$.

To prove Eqn.~(\ref{eqn:ag_rescale1}), suppose that $\tilde{X}$ and $\tilde{\cal S}$ give the maximal interference level in ${\bf I}_{\max, \sf d}^{\sf cpcs}[{\sf N} + 1, \alpha, {\sf P}, {\sf N}]$. Since $\tilde{X}$ lies in $\Re^{\sf d}$, we can rescale the distances between any pair of transmitters in $\tilde{X}$ by a factor of $({\sf t}_{\rm cs} - {\sf N})^{-1/\alpha}$. We denote such a set of transmitters after rescaling as $\tilde{X'}$. Namely, 
\begin{equation}
|t'_j - t'_i| = |{\sf t}_{\rm cs} - {\sf N}|^{-1/\alpha} |t_j - t_i| 
\end{equation}
where $t'_j, t'_i \in X'$, $t_j, t_i \in X$.
It is evident that 
\begin{equation}
\begin{array}{@{}r@{\ }l@{}} 
& {\sf N} + \underset{i_j \in \{i_1, ..., i_{k-1}\}}{\sum} {\sf P}
|t_{i_j} - t_{i_k}|^{-\alpha} \le {\sf N} + 1 \\
\Leftrightarrow &
{\sf N} + \underset{i'_j \in \{i'_1, ..., i'_{k-1}\}}{\sum} {\sf P}
|t_{i'_j} - t_{i'_k}|^{-\alpha} \le {\sf t}_{\rm cs}
\end{array}
\end{equation}
where $(i_1, ..., i_{k}) = \tilde{\cal S}$ and $(i'_1, ..., i'_{k}) = \tilde{\cal S'}$, and $\tilde{\cal S'}$ is the induced sequence of $\tilde{\cal S}$ in the rescaled $\tilde{X'}$.

Hence, the induced $\tilde{\cal S'}$ must also give the maximal interference level for ${\bf I}_{\max, \sf d}^{\sf cpcs}[{\sf t}_{\rm cs}, \alpha, {\sf P, N}]$. Therefore, 
\begin{equation} 
\begin{array}{@{}r@{\ }l@{}}
& {\bf I}_{\max, \sf d}^{\sf cpcs}[{\sf t}_{\rm cs}, \alpha, {\sf P, N}] = \underset{j \in \tilde{S'}}{\sum} |t'_j - t'_i|^{-\alpha} \\
= &|{\sf t}_{\rm cs} - {\sf N}| \cdot \underset{j \in \tilde{\cal S}}{\sum} |t_j - t_i|^{-\alpha} \\
=  & |{\sf t}_{\rm cs} - {\sf N}| \cdot {\bf I}_{\max, \sf d}^{\sf cpcs}[{\sf N} + 1, \alpha, {\sf P, N}]
 \end{array}
\end{equation}

To prove Eqn.~(\ref{eqn:ag_rescale2}), it follows a similar approach. However, we multiply the distances between all transmitters in $\tilde{X}'$ by a factor of ${\sf P}^{1/\alpha}$.
\end{proof}

\medskip

\usecounter{theorem}
\setcounter{theorem}{0}  
\begin{theorem} 
Given a set of links $X$, which lies in Euclidean space $\Re^{\sf d}$, let $d_{\max} \triangleq \max_{i \in X} |t_i - r_i|$.
If we set 
\begin{equation}
{\sf t}_{\rm cs} \le  {\sf P} \Big( 2 d_{\max} +   
\big( \frac{1}{{\bf I}_{\max, \sf d}^{\sf cpcs}[\alpha]}
(\frac{d_{\max}^{-\alpha}}{\beta} - \frac{\sf N}{\sf P}) \big)^{\frac{-1}{\alpha}}
\Big)^{-\alpha}  + {\sf N} 
\end{equation}
then it can ensure interference-safe transmissions in CPCS, namely,
\begin{equation}
{\cal S} \in {\mathscr C}^{\sf cpcs}_{\sf P, N}\big[{X},{\sf t}_{\rm cs}\big] \ \ \Rightarrow \ \ {\cal S} \in {\mathscr B}^{}_{\sf P, N}\big[{X},\beta\big] 
\end{equation}
\end{theorem}

\begin{proof} 
Suppose that ${\cal S} \in {\mathscr C}^{\sf cpcs}_{\sf P, N}\big[{X},{\sf t}_{\rm cs}\big]$, and ${\sf t}'_{\rm cs}$ satisfies $\big( \frac{|{\sf t}_{\rm cs} - {\sf N}|}{{\sf P}} \big)^{\frac{-1}{\alpha}} \ge \big( \frac{|{\sf t}'_{\rm cs} - {\sf N}|}{{\sf P}} \big)^{\frac{-1}{\alpha}} +  2 d_{\max}$.

Then, by Lemmas~\ref{lem:bi_intf} and \ref{lem:ag_rescale}, we obtain: 
\begin{eqnarray}
& & \displaystyle \frac{{\sf P} |t_i - r_i|^{-\alpha}}{{\sf N} + \underset{j \in {\cal S} \backslash \{ i \}}{\sum} {\sf P} \cdot {\sf dist}(i,j)^{-\alpha}}  \\ 
& \ge & \displaystyle
  \frac{{\sf P} d_{\max}^{-\alpha}}{{\sf N} + {\sf P}\cdot{\bf B}_{\max, \sf  d}^{\sf cpcs}[{\sf t}_{\rm cs}, \alpha, {\sf P}, {\sf N}]} \\
& \ge & \displaystyle  \frac{{\sf P}d_{\max}^{-\alpha}}{{\sf N} + {\sf P}\cdot{\bf I}_{\max, \sf d}^{\sf cpcs}[{\sf t}'_{\rm cs}, \alpha, {\sf P}, {\sf N}]} \\
  & =  & \displaystyle \frac{{\sf P}\cdot d_{\max}^{-\alpha}}{{\sf N} + (|{\sf t}'_{\rm cs} - {\sf N}|) \cdot {\bf I}_{\max, \sf d}^{\sf cpcs}[{\sf N} + 1, \alpha, {\sf P}, {\sf N}] } 
\end{eqnarray}
We note that ${\mathscr C}^{\sf cpcs}_{\sf P, N}[{X},{\sf N} + 1] = {\mathscr C}^{\sf cpcs}_{{\sf P},{\sf N}=0}[{X},1]$, and hence,
\begin{equation} 
{\bf I}_{\max, \sf d}^{\sf cpcs}[{\sf N} + 1, \alpha, {\sf P}, {\sf N}] = {\bf I}_{\max, \sf d}^{\sf cpcs}[1, \alpha, {\sf P}, 0]
\end{equation}
Thus,
\begin{equation} 
 \frac{{\sf P} |t_i - r_i|^{-\alpha}}{{\sf N} + \underset{j \in {\cal S} \backslash \{ i \}}{\sum} {\sf P} \cdot {\sf dist}(i,j)^{-\alpha}} 
\ge \frac{{\sf P}\cdot d_{\max}^{-\alpha}}{{\sf N} + (|{\sf t}'_{\rm cs} - {\sf N}|) \cdot {\bf I}_{\max, \sf d}^{\sf cpcs}[\alpha] }
\end{equation}
Therefore, if $\frac{{\sf P}\cdot d_{\max}^{-\alpha}}{{\sf N} + (|{\sf t}'_{\rm cs} - {\sf N}|) \cdot {\bf I}_{\max, \sf d}^{\sf cpcs}[\alpha] } \ge \beta$, then
\begin{equation} 
\frac{{\sf P} |t_i - r_i|^{-\alpha}}{{\sf N} + \underset{j \in {\cal S} \backslash \{ i \}}{\sum} {\sf P} \cdot {\sf dist}(i,j)^{-\alpha}} 
\ge \frac{{\sf P}\cdot d_{\max}^{-\alpha}}{{\sf N} + (|{\sf t}'_{\rm cs} - {\sf N}|) \cdot {\bf I}_{\max, \sf d}^{\sf cpcs}[\alpha] } \ge \beta
\end{equation}
Namely, ${\cal S} \in {\mathscr B}^{}_{\sf P, N}\big[{X},\beta\big]$.

Note that $\frac{{\sf P}\cdot d_{\max}^{-\alpha}}{{\sf N} + (|{\sf t}'_{\rm cs} - {\sf N}|) \cdot {\bf I}_{\max, \sf d}^{\sf cpcs}[\alpha] } \ge \beta$ is equivalent to
\begin{equation} 
2 d_{\max} +   
\big( \frac{1}{{\bf I}_{\max, \sf d}^{\sf cpcs}[\alpha]}
(\frac{d_{\max}^{-\alpha}}{\beta} - \frac{\sf N}{\sf P}) \big)^{\frac{-1}{\alpha}} \le \big( \frac{|{\sf t}'_{\rm cs} - {\sf N}|}{{\sf P}} \big)^{\frac{-1}{\alpha}} +  2 d_{\max}
\end{equation}

Since $\big( \frac{|{\sf t}_{\rm cs} - {\sf N}|}{{\sf P}} \big)^{\frac{-1}{\alpha}} \ge \big( \frac{|{\sf t}'_{\rm cs} - {\sf N}|}{{\sf P}} \big)^{\frac{-1}{\alpha}} +  2 d_{\max}$, we obtain the corresponding setting of ${\sf t}_{\rm cs}$ for ensuring interference-safe transmissions (i.e., ${\cal S} \in {\mathscr B}^{}_{\sf P, N}\big[{X},\beta\big]$) by:
\begin{equation}
{\sf t}_{\rm cs} \le  {\sf P} \Big( 2 d_{\max} +   
\big( \frac{1}{{\bf I}_{\max, \sf d}^{\sf cpcs}[\alpha]}
(\frac{d_{\max}^{-\alpha}}{\beta} - \frac{\sf N}{\sf P}) \big)^{\frac{-1}{\alpha}}
\Big)^{-\alpha}  + {\sf N} 
\end{equation}
\end{proof}

\subsection{1-D Case}

We can evaluate ${\bf I}_{\max, 1}^{\sf cpcs}[\alpha]$ by considering the closest packing on the real line $\Re$.
We let $d_i \triangleq |t_{2i-1} - t_{2i-3}|$ and $c_i \triangleq |t_{2i} - t_{2i-2}|$.

The closest packing can be obtained as follows. First, place $t_1$ as closely as possible to $t_0$, such that $|t_1-t_0|^{-\alpha} = 1$. Then, place $t_2$ as closely as possible to $t_0$, such that $|t_2-t_0|^{-\alpha} + |t_2-t_1|^{-\alpha} = 1$. We iterate this process for the $k$-th transmitter $t_k$ as in Fig.~\ref{fig:realine}.

\begin{figure}[htb!]  \vspace{-5pt} 
    \centering
    \includegraphics[scale=0.7]{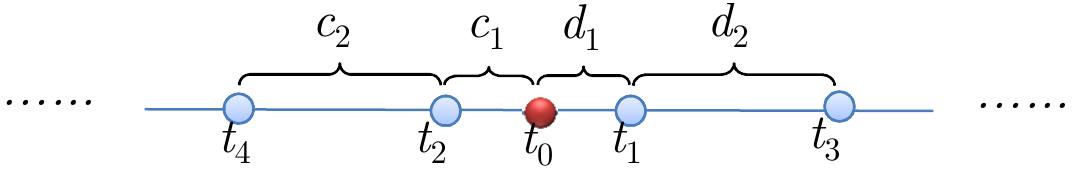}
  \caption{The closest packing in 1-D Case to attain ${\bf I}_{\max, 1}^{\sf cpcs}[\alpha]$.} \label{fig:realine}
\end{figure}

Hence, we obtain
\begin{equation} \hspace{-5pt}
\begin{array}{@{}r@{\ }l@{}}
1 = & (d_1)^{-\alpha} \  \Rightarrow \ d_1 = 1\\
1 = &  (c_1)^{-\alpha} + (c_1 + d_1)^{-\alpha} \\
1 = &  (d_2)^{-\alpha} + (d_1 + d_2)^{-\alpha} + (c_1 + d_1 + d_2)^{-\alpha} \\
1 = & (c_2)^{-\alpha} + (c_2 + c_1)^{-\alpha} + (c_2 + c_1 + d_1)^{-\alpha} \\
& + (c_2 + c_1 + d_1 + d_2)^{-\alpha} \\
& \qquad \qquad \vdots   \\
\end{array} \notag
\end{equation}
\begin{equation} \hspace{-5pt}
\begin{array}{@{}r@{\ }l@{}}
1 = & (d_k)^{-\alpha} + (d_{k-1} + d_k )^{-\alpha} + ... + (d_1 + ... + d_k)^{-\alpha} \\
& + (c_1 + d_1 + ... + d_k)^{-\alpha} + ... \\
 & + (c_{k-1} + ... + c_1 + d_1 + ... + d_k)^{-\alpha} \\
1 = & (c_k)^{-\alpha} + (c_k + c_{k-1} )^{-\alpha} + ... + (c_k + ... + c_1)^{-\alpha} \\
& + (c_k + ... + c_1 + d_1)^{-\alpha}
+ ... \\
& + (c_k + ... + c_1 + d_1 + ... + d_k)^{-\alpha} \\
\end{array}
\end{equation}

However, it is difficult to obtain $\{ c_i, d_i\}_{i = 1, 2, ...}$ for general $\alpha$. Next, we obtain a upper bound for ${\bf I}_{\max, 1}^{\sf cpcs}[\alpha]$.
First, we note that
\begin{equation*} 
d_1 < c_1 < d_2 < c_2 < ... < d_{k-1} < c_{k-1} < d_k < c_k < ...
\end{equation*}
Hence, we obtain
\begin{equation}
\begin{array}{@{}r@{\ }l@{}}
c_1 + d_1 & < c_1 + c_1 \\
(c_1 + d_1)^{-\alpha} & > (c_1 + c_1)^{-\alpha} \\
1 = (c_1)^{-\alpha} + (c_1 + d_1)^{-\alpha} & > (c_1)^{-\alpha} (1 + 2^{-\alpha}) \\
\Rightarrow\  c_1 & > \frac{1}{ (1 + 2^{-\alpha})^{-1/\alpha}}
\end{array}
\end{equation}
Similarly, we obtain
\begin{equation*} 
\begin{array}{@{}r@{\ }l@{}}
1 = & (d_k)^{-\alpha} + (d_{k-1} + d_k )^{-\alpha} + ... + (d_1 + ... + d_k)^{-\alpha} \\
& + (c_1 + d_1 + ... + d_k)^{-\alpha}  + ...  \\
& + (c_{k-1} + ... + c_1 + d_1 + ... + d_k)^{-\alpha} \\
   > & (d_k)^{-\alpha} ( 1 + 2^{-\alpha} + ... + (2k-1)^{-\alpha})
\end{array}
\end{equation*}
\begin{equation} \label{eqn:dk}
\Rightarrow \quad d_k  > \Big(\sum_{i=1}^{2k-1} i^{-\alpha} \Big)^{1/\alpha}
\end{equation}
\begin{equation*}
\begin{array}{@{}r@{\ }l@{}}
1 = & (c_k)^{-\alpha} + (c_k + c_{k-1} )^{-\alpha} + ... + (c_k + ... + c_1)^{-\alpha} \\
& + (c_k + ... + c_1 + d_1)^{-\alpha}  + ... \\
& + (c_k + ... + c_1 + d_1 + ... + d_k)^{-\alpha} \\
  > & (c_k)^{-\alpha} ( 1 + 2^{-\alpha} + ... + (2k)^{-\alpha}) \\
\end{array}
\end{equation*}
\begin{equation}
\Rightarrow \quad  c_k  > \Big(\sum_{i=1}^{2k} i^{-\alpha} \Big)^{1/\alpha}
\end{equation}

Therefore,
\begin{equation*} 
\begin{array}{@{}r@{\ }l@{}}
{\bf I}_{\max, 1}^{\sf cpcs}[\alpha]= & \displaystyle \sum_{n = 1}^{\infty} (c_1 + ... + c_n)^{-\alpha} +  \sum_{n = 1}^{\infty} (d_1 + ... + d_n)^{-\alpha} \\
< &  \displaystyle \bar{\bf I}_1[\alpha] \triangleq \sum_{n = 1}^{\infty} \Big( \sum_{k = 1}^{n} \Big(\sum_{i=1}^{2k} i^{-\alpha} \Big)^{1/\alpha} \Big)^{-\alpha} \\
& \qquad \quad  \displaystyle + \sum_{n = 1}^{\infty} \Big( \sum_{k = 1}^{n} \Big(\sum_{i=1}^{2k-1} i^{-\alpha} \Big)^{1/\alpha} \Big)^{-\alpha}   
\end{array}
\end{equation*}

Numerically, we evaluate $\bar{\bf I}_1[\alpha]$ in Fig.~\ref{fig:1d_I} by summing only the first $n$ terms of $\bar{\bf I}_1[\alpha]$. We observe that $\bar{\bf I}_1[\alpha]$ converges quickly (see Fig.~\ref{fig:1d_I}). 
We tabulate the values of $\bar{\bf I}_1[\alpha]$ via numerical study in Table~\ref{tab:I_1}. 
Note that ${\bf I}_{\max, 1}^{\sf cpcs}[2] \approx 2.59$ by numerical study. Hence, the upper bound appears to be tight.

\begin{figure}[htb!]  
    \centering
    \includegraphics[scale=0.75]{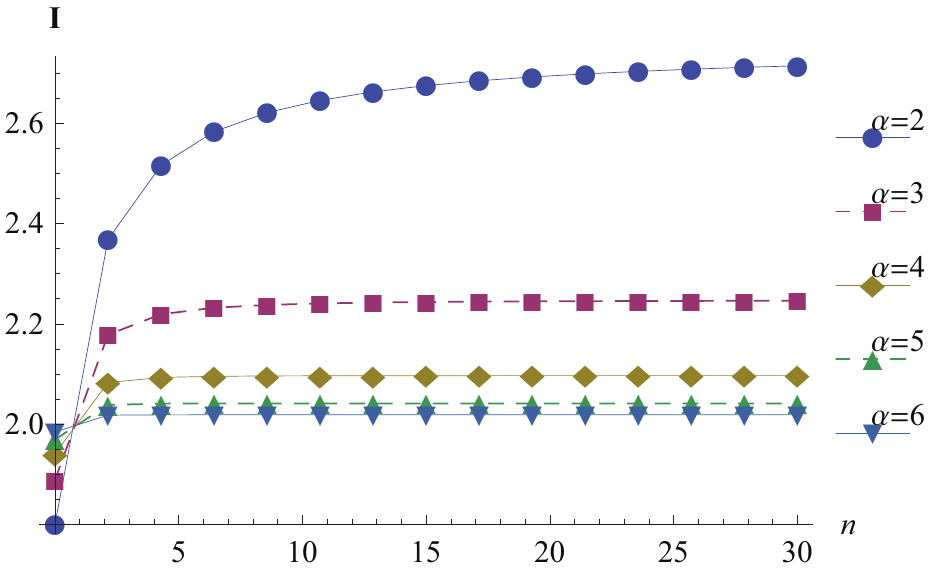} 
   \caption{Numerical values of $\bar{\bf I}_1[\alpha]$ of the first $n$ terms in the summation.} \label{fig:1d_I}  
\end{figure}

\begin{table}[ht] \centering 
\caption{Numerical values of $\bar{\bf I}_1[\alpha]$.}   \label{tab:I_1}
\begin{tabular}{@{}c@{\ }|@{\ }c@{\ }|@{\ }c@{\ }|@{\ }c@{\ }|@{\ }c@{\ }|@{\ }c@{}}  \hline \hline
& \multirow{2}{*}{$\alpha=2$} & \multirow{2}{*}{$\alpha=3$} & \multirow{2}{*}{$\alpha=4$} & \multirow{2}{*}{$\alpha=5$} & \multirow{2}{*}{$\alpha=6$}\\ 
& & & & & \\
\hline
\multirow{2}{*}{$\bar{\bf I}_1[\alpha]$} & \multirow{2}{*}{2.74438} & \multirow{2}{*}{2.24708} & \multirow{2}{*}{2.09705} & \multirow{2}{*}{2.04166} & \multirow{2}{*}{2.01887}\\
& & & & & \\
\hline \hline
\end{tabular} 
\end{table}

\subsection{2-D Case}

Obtaining the maximal interference level ${\bf I}_{\max, 2}^{\sf cpcs}[\alpha]$ for Theorem~\ref{thm:ag_imax} in 2-D case is more complicated, as there are many more possible locations of nodes in $\Re^2$. 

We next give an upper bound for ${\bf I}_{\max, 2}^{\sf cpcs}[\alpha]$. Recall the sequence of separation distances $(d_1, d_2, ...)$ from the definition in Fig.~\ref{fig:realine} for the 1-D case. We set the spacing distance in the hexagonal grid to be $d_1=1$. 

Nodes are placed as hexagonal rings around $t_0$ (see Fig.~\ref{fig:hex2}). We denote the set of nodes in hexagonal grid for the $i$-th ring by ${\cal H}_i = \{t_i^1, ..., t_i^{|{\cal H}_i|}\}$. Particularly, $t_i^1$ are placed on the positive horizontal real line in $\Re^2$.
We set the separation distance between the rings according to sequence $(d_1, d_2, ...)$. Namely, the location of $t_i^1$ in $\Re^2$ is $\big( \lfloor \sum_{j=1}^{i} d_i \rfloor, 0 \big)$.

We can upper bound ${\bf I}_{t_0}[\bigcup_{j=1}^{\infty} {\cal H}_j]$ from Fig.~\ref{fig:hex2}. For each $i$-th ring, $|{\cal H}_i| = 6 \big( \lfloor\sum_{j=1}^{i} d_i\rfloor\big)$. Therefore,   
\begin{equation*} 
\begin{array}{@{}r@{\ }l@{}}
{\bf I}_{\max, 2}^{\sf cpcs}[\alpha]< & {\bf I}_{t_0}[\bigcup_{j=1}^{\infty} {\cal H}_j] <  \displaystyle \sum_{n = 1}^{\infty} 6 (d_1 + ... + d_n)^{-\alpha+1} \\
< &  \displaystyle \bar{\bf I}_2[\alpha] \triangleq  \displaystyle 6 \sum_{n = 1}^{\infty} \Big( \sum_{k = 1}^{n} \Big(\sum_{i=1}^{2k-1} i^{-\alpha} \Big)^{1/\alpha} \Big)^{-\alpha+1} 
\end{array}
\end{equation*}

\begin{figure}[htb!]
      \centering
    \includegraphics[scale=0.7]{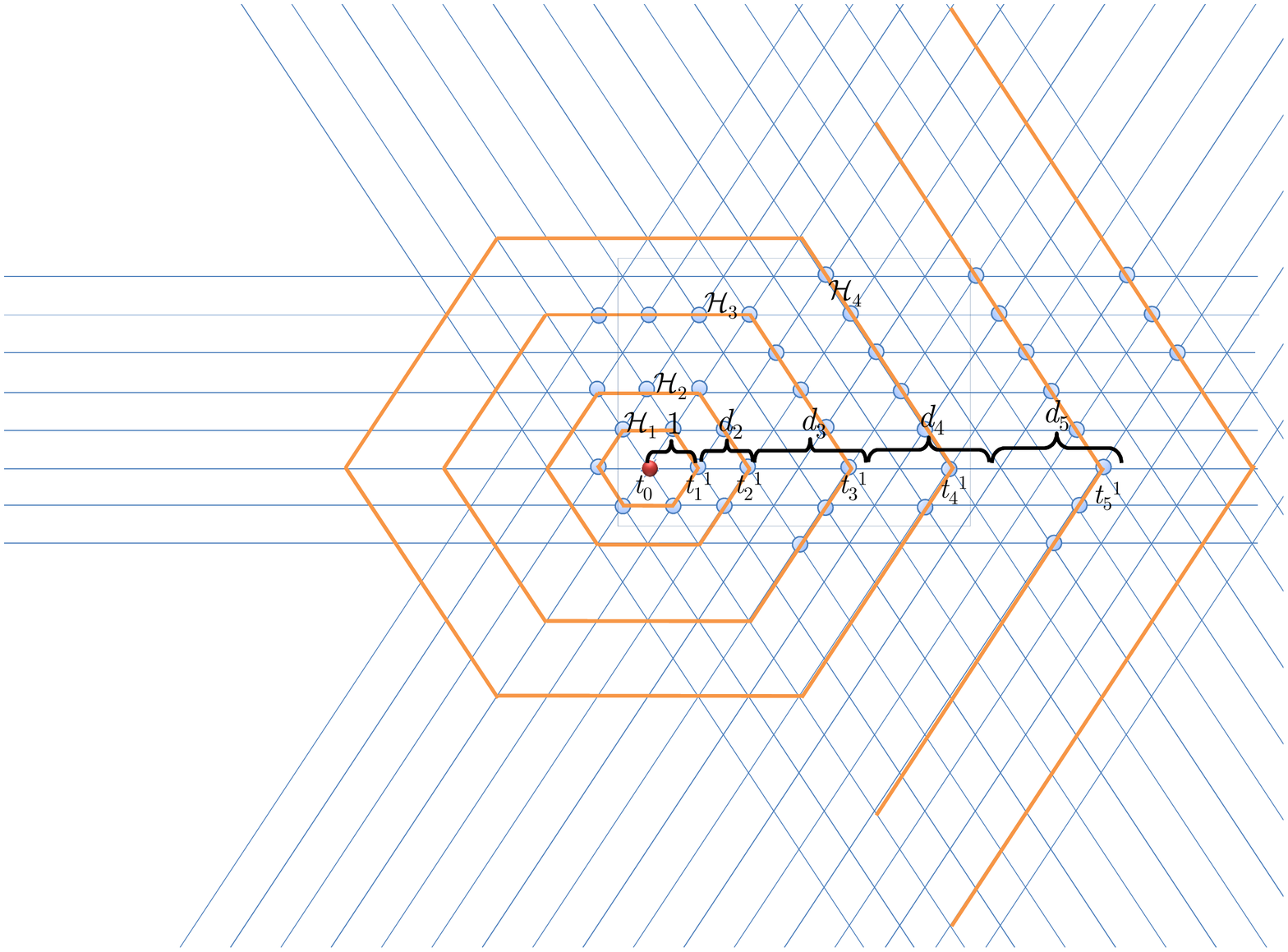} 
  \caption{Nodes are placed as hexagonal rings, separated by $d_i$.} \label{fig:hex2} 
\end{figure}

Numerically, we evaluate $\bar{\bf I}_2[\alpha]$ in Fig.~\ref{fig:2d_I} by summing only the first $n$ terms in the outmost summation of $\bar{\bf I}_2[\alpha]$. We observe that $\bar{\bf I}_1[\alpha]$ converges quickly as $n$ increases. We tabulate the values of $\bar{\bf I}_2[\alpha]$ via numerical study in Table~\ref{tab:I_2}.

\begin{figure}[htb!] 
    \centering
    \includegraphics[scale=0.75]{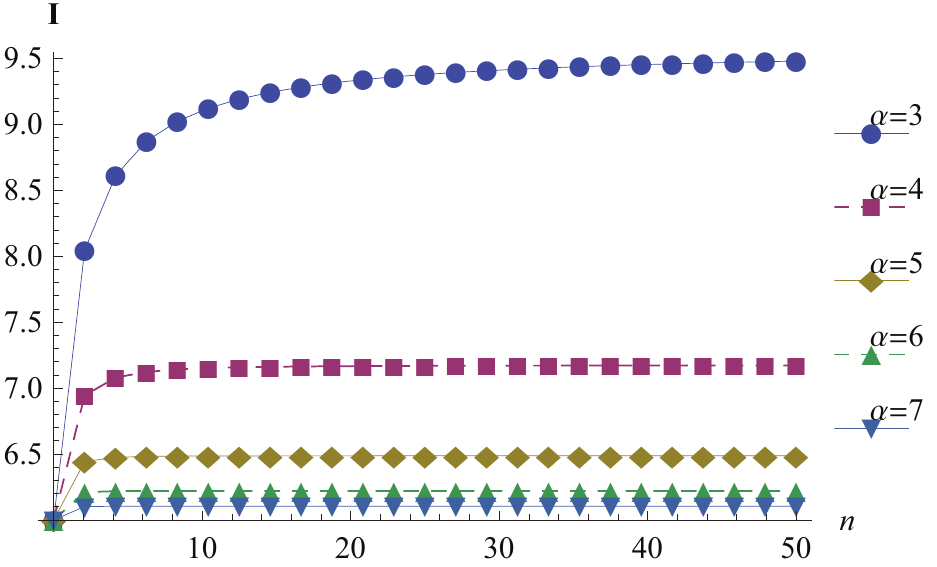} 
   \caption{Numerical values of $\bar{\bf I}_2[\alpha]$ of the first $n$ terms in the summation.} \label{fig:2d_I} 
\end{figure}

\begin{table}[ht] \centering 
\caption{Numerical values of $\bar{\bf I}_2[\alpha]$.}  \label{tab:I_2}
\begin{tabular}{@{}c@{\ }|@{\ }c@{\ }|@{\ }c@{\ }|@{\ }c@{\ }|@{\ }c@{\ }|@{\ }c@{}}  \hline \hline
& \multirow{2}{*}{$\alpha=3$} & \multirow{2}{*}{$\alpha=4$} & \multirow{2}{*}{$\alpha=5$} & \multirow{2}{*}{$\alpha=6$} & \multirow{2}{*}{$\alpha=7$}\\ 
& & & & & \\
\hline
\multirow{2}{*}{$\bar{\bf I}_2[\alpha]$} & 
\multirow{2}{*}{9.56077} & 
\multirow{2}{*}{7.17297} & 
\multirow{2}{*}{6.48636} & 
\multirow{2}{*}{6.21992} & 
\multirow{2}{*}{6.10368} \\
& & & & & \\
\hline \hline
\end{tabular} 
\end{table}

\end{document}